\documentclass[a4paper,UKenglish,cleveref, autoref, thm-restate]{lipics-v2021}

\usepackage[composition,notion, quotation,xcolor]{knowledge}

\usepackage{perso}

\usepackage{mathtools}
\usepackage{xspace}

\usepackage{booktabs}

\usepackage{tikz,xcolor,hyperref}
\usetikzlibrary{positioning, arrows}
\usetikzlibrary{patterns}
\usetikzlibrary{shapes}
\usetikzlibrary{decorations.pathreplacing}
\usetikzlibrary{fadings}
\usetikzlibrary{intersections}

\usepackage{lineno}

\knowledgenewrobustcmd{\Value}[1]{\cmdkl{\mathsf{Val}_{#1}}}

\knowledgenewrobustcmd{\play}{\cmdkl{\mathsf{play}}}
\knowledgenewrobustcmd{\wt}{\cmdkl{wt}}

\newcommand{\reset}[1]{{#1}:=0}

\knowledgenewrobustcmd{\Reg}[1]{\cmdkl{\mathsf{Reg}_{#1}}}
\knowledgenewrobustcmd{\Reglt}[1]{\cmdkl{\mathsf{Reg}_{#1}^<}}

\knowledgenewrobustcmd{\DDelta}{\cmdkl{\Delta}}
\knowledgenewrobustcmd{\kkappa}{\cmdkl{\kappa}}

\newcommand{\Maxsf}{\mathsf{Max}}
\newcommand{\Minsf}{\mathsf{Min}}
\newcommand{\Mamax}{\Maxsf}
\newcommand{\Mimin}{\Minsf}
\newcommand{\Player}{\mathsf{P}}
\knowledgenewrobustcmd{\upclock}[1]{\cmdkl{X_{#1}^\uparrow}}
\newcommand{\weight}{\mathsf{weight}}
\knowledgenewrobustcmd{\opt}[2]{\cmdkl{\mathsf{opt}^{#1}_{#2}}}
\knowledgenewrobustcmd{\Opt}[2]{\cmdkl{\mathsf{Opt}^{#1}_{#2}}}
\knowledgenewrobustcmd{\Next}{\cmdkl{\mathsf{Leaving}}}
\knowledgenewrobustcmd{\reg}[1]{\cmdkl{\mathsf{reg}(#1)}}
\knowledgenewrobustcmd{\adh}[1]{\cmdkl{\overline{#1}}}
\newcommand{\regadh}[1]{\adh{\mathsf{reg}(#1)}}
\newcommand{\wout}{w_{\mathsf{out}}}
\newcommand{\wtg}{\textsc{wtg}\xspace}

\newcommand{\wtgs}{\textsc{wtg}s\xspace}

\knowledgenewrobustcmd{\RC}[1]{\cmdkl{\Rcal(#1)}}
\knowledgenewrobustcmd{\TC}[1]{\cmdkl{\Tcal(#1)}}
\knowledgenewrobustcmd{\Gtilde}{\cmdkl{\tilde{\Gcal}}}

\newcommand{\openquote}{``}

\newcommand{\ie}{i.e.,~}
\newcommand{\etal}{\textit{et al.}~}

\usetikzlibrary{decorations,decorations.markings} 
\pgfkeys{/tikz/.cd,
    contour distance/.store in=\ContourDistance,
    contour distance=-10pt, 
    contour step/.store in=\ContourStep,
    contour step=1pt,
}

\pgfdeclaredecoration{closed contour}{initial}
{%
\state{initial}[width=\ContourStep,next state=cont] {
    \pgfmoveto{\pgfpoint{\ContourStep}{\ContourDistance}}
    \pgfcoordinate{first}{\pgfpoint{\ContourStep}{\ContourDistance}}
    \pgfpathlineto{\pgfpoint{0.3\pgflinewidth}{\ContourDistance}}
    \pgfcoordinate{lastup}{\pgfpoint{1pt}{\ContourDistance}}
    
  }
  \state{cont}[width=\ContourStep]{
     \pgfmoveto{\pgfpointanchor{lastup}{center}}
     \pgfpathlineto{\pgfpoint{\ContourStep}{\ContourDistance}}
     \pgfcoordinate{lastup}{\pgfpoint{\ContourStep}{\ContourDistance}}
  }
  \state{final}[width=\ContourStep]
  { 
    \pgfmoveto{\pgfpointanchor{lastup}{center}}
    \pgfpathlineto{\pgfpointanchor{first}{center}}
  }
}

\newcommand{\eqdef}{\stackrel{\mathsf{def}}{=}}

\newcommand{\set}[1]{\left\{\; {#1} \;\right\}}

\newcommand{\setof}[2]{\left\{{#1} \;:\; {#2} \right\}}

\newcommand{\Nat}{\mathbb{N}}
\newcommand{\Rel}{\mathbb{Z}}
\newcommand{\Reel}{\mathbb{R}}

\newcommand{\DefineNewOrder}[4]{%
	\expandafter\newcommand\csname#1le\endcsname{\mathrel{#2}_{#4}}%
	\expandafter\newcommand\csname#1ge\endcsname{\mathrel{\reflectbox{\text{$#2$}}}_{#4}}%
	\expandafter\newcommand\csname#1leq\endcsname{\mathrel{#3}_{#4}}%
	\expandafter\newcommand\csname#1geq\endcsname{\mathrel{\text{\reflectbox{$#3$}}}_{#4}}%
	\expandafter\newcommand\csname#1eq\endcsname{\mathrel{\equiv}_{#4}}%
	\expandafter\newcommand\csname#1perp\endcsname{\mathrel{\perp}_{#4}}%
}

\knowledgenewrobustcmd{\CCE}{\cmdkl{\text{CEC}}}
\knowledgenewrobustcmd{\CZ}{\cmdkl{\text{ZC}}}
\knowledgenewrobustcmd{\CNZ}{\cmdkl{\text{NZC}}}
\knowledgenewrobustcmd{\CM}{\cmdkl{\text{MC}}}

\knowledgenewrobustcmd{\Mcal}{\cmdkl{\mathcal{M}}}
\knowledgenewrobustcmd{\GMcal}{\cmdkl{\Gcal_{\mathcal{M}}}}

\knowledge{notion}
 | Clock constraints
 | clock constraints
 | guards
 | guard
\knowledge{notion}
 | weighted timed game
 | \wtg 
 | \wtgs

\knowledge{notion}
 | almost non-Zeno \wtg 
 | almost non-Zeno \wtgs 
 | almost non-Zeno
\knowledge{notion}
 | output edges
 | output edge
\knowledge{notion}
 | Value Problem
 | value

\knowledge{notion}
 | kernel
 | kernels
\knowledge{notion}
 | green
 | green location
 | \emph {green} transition
 | green transitions
 | fully-green
\knowledge{notion}
 | kernel weighted timed games
 | kernel weighted timed game
 | kernel \wtg 
 | kernel \wtgs 
 | kernel game
\knowledge{notion}
 | region construction
\knowledge{notion}
 | output weight function
 | output weight functions
 | output function
 | output functions
\knowledge{notion}
 | value iteration algorithm
\knowledge{notion}
 | regions
 | region
\knowledge{notion}
 | $\intfo 01$-region
 | $\intfo 01$-regions
\knowledge{notion}
 | \wtg with clock values in $\intfo 01$
 | $\intfo 01$
 | $\intfo 01$-
 | $\intfo 01$-\wtgs
 | $\intfo 01$-\wtg 
\knowledge{notion}
 | $\intff 01$-
 | $\intff 01$-\wtg 
\knowledge{notion}
 | projected output functions
 | projected output function

\knowledge{notion}
 | \textbf {region} \wtg 
 | region \wtg 
 | region \wtgs 
 | region
 | a region weighted timed game
\knowledge{notion}
 | sub-\wtg
 | sub-\wtgs

\knowledge{notion}
 | semi-unfolding
 | semi-unfolding construction
 | partial unfolding
 | semi-unfolded \wtg 
\knowledge{notion}
 | time-successor
\knowledge{notion}
 | circular clock difference
 | \textbf {circular clock difference}
\knowledge{notion}
 | trimmed
\knowledge{notion}
 | relaxed
 
\knowledge{notion}
 | region assignment
\knowledge{notion}
 | adherence
 
\knowledge{notion}
 | induce continuous weight
 | inducing continuous weight
 | inducing a continuous weight
 | induces continuous weight
 
\knowledge{notion}
  | deleting information

\knowledge{notion}
 | almost divergent \wtg 
 | almost divergent
 | \textbf {almost divergent}
 | almost divergent \wtgs 
  | almost divergence
\knowledge{notion} 
 | two-counter machines
 | two-counter machine

\knowledge{notion}
  | zero-cycle
  | zero-cycles

\title{(Un)Decidability of Two-Clock Weighted Timed Games}

\author{Quentin {Guilmant}}{
Max Planck Institute for Software Systems, Saarland Informatics Campus, Germany}{Quentin.Guilmant@irif.fr}{}{}

\author{Joël {Ouaknine}}{Max Planck Institute for Software Systems, Saarland Informatics Campus, Germany}{joel@mpi-sws.org}{}{}

\author{Isa {Vialard}}{Max Planck Institute for Software Systems, Saarland Informatics Campus, Germany \and \url{https://isavialard.github.io/home/} }{vialard@mpi-sws.org}{https://orcid.org/0000-0002-7261-9342}{}

\author{Angel {Wuttke}}{École Normale Supérieure Paris-Saclay, Gif-sur-Yvette, France\\
Max Planck Institute for Software Systems, Saarland Informatics Campus, Germany}{angel.wuttke@ens-paris-saclay.fr}{}{}

\authorrunning{Q. Guilmant, J. Ouaknine, I. Vialard, A. Wuttke} 

\Copyright{Quentin Guilmant, Joël Ouaknine, Isa Vialard, Angel Wuttke} 



\category{} 

\nolinenumbers 

\begin{document}

\maketitle

\begin{abstract}
 The Value Problem for weighted timed games consists in determining, given a two-player weighted timed game with a reachability objective and a rational threshold, whether or not the value of the game exceeds the threshold.
This problem is known to be undecidable for weighted timed games with three or more clocks, and decidable for one-clock weighted timed games.

In this paper,  we settle the remaining two-clock case,
 which has remained stubbornly open for a decade. We show that two-clock weighted timed games lie precisely at the decidability frontier: On the one hand the Value problem is undecidable for two-clock weighted timed games, even when restricted to non-negative weights and under a time bound; On the other hand, we establish decidability under an almost divergence restriction, by providing a $3$\textsc{ExpTime} algorithm for computing the value of almost divergent turn-based two-clock weighted timed games.
\end{abstract}
%
%

\tableofcontents

\definecolor{myblue}{HTML}{BBEEFF}
\definecolor{mygreen}{HTML}{CCFFCC}
\definecolor{myyellow}{HTML}{FFFFBB}
\definecolor{mypurple}{HTML}{FECCFC}

\newcommand{\quentin}[1]{\todo[color=myblue]{#1}}
\newcommand{\quentininline}[1]{\todo[inline,color=myblue]{#1}}
\newcommand{\isa}[1]{\todo[color=mygreen]{#1}}
\newcommand{\isainline}[1]{\todo[inline,color=mygreen]{#1}}
\newcommand{\joel}[1]{\todo[color=mypurple]{#1}}
\newcommand{\joelinline}[1]{\todo[inline,color=mypurple]{#1}}
\newcommand{\angel}[1]{\todo[color=myyellow]{#1}}
\newcommand{\angelinline}[1]{\todo[inline,color=myyellow]{#1}}

\section{Introduction}

Timed automata, introduced by Alur and Dill~\cite{Alur94}, are a fundamental model for the verification
of real-time systems. They extend finite automata with a finite number of real-valued clocks,
whose values govern the enabling of transitions, while time can elapse continuously within
locations. 

Weighted timed games (\wtgs) are timed automata equipped with integer
weights on locations and transitions, and a game structure. The weights on transitions are
instantaneous, while location weights are linear in the delay spent in the location. Locations are
partitioned between player Min, who seeks to minimize the cost required to reach a target
location, and player Max, who represents an uncontrollable environment.

Weighted timed games offer a powerful framework for the formal verification of real-time systems, allowing temporal constraints to be considered alongside quantitative objectives such as task scheduling, energy management, resource allocation, or performance optimization.

A natural problem on weighting timed games is the following: 
Given a \wtg, can we compute its Value, \ie the infimum of the optimal cost\footnote{where the optimal cost is the supremum over all possible strategies of $\Maxsf$ of the weight of the path produced by the strategy profile.}
 over all strategies of $\Minsf$? Or, formulated as a decision problem:
is its Value less than or equal to some given threshold $c$?

 This is the \emph{""Value Problem""}, not to be confused with the \emph{Existence Problem}:
 Given a \wtg and a threshold $c$, does $\Minsf$ have a strategy to reach her goal location with cost at most $c$? 
These two problems can yield different answers (see \cref{fig-stupid}).
\begin{figure}
\centering
\input{figures/WTA/WTA}
\caption{Example of a \wtg where the Value is $0$, but every strategy of $\Minsf$ yields cost $>0$.}
\label{fig-stupid}
\end{figure}

%

In this paper, we focus on the "Value Problem". While decidable for weighted timed automata, the "Value Problem" is undecidable for weighted timed games (\cite{Bouyer2015}).

However, one can recover decidability by restricting the number of clocks. Bouyer \etal~\cite{bouyer2006} established that the "Value Problem" is decidable for one-clock \wtg with non-negative weights. Monmege, Parreaux and Reynier~\cite{monmege2024journal} extended this decidability result to one-clock \wtg with arbitrary weights in .
On the other hand, Bouyer,Jaziri and Markey~\cite{Bouyer2015} prove undecidability for \wtgs with non-negative weights and three or more clocks.

Another way to recover decidability is to impose non-Zeno, or more generally divergence, conditions.
Named after Zeno's paradox, a run exhibits Zeno-like behaviour if it is infinite while accumulating only a finite amount of resources (be it time or, for weighted timed automata, cost); phenomena of this kind play a central role in the undecidability of the Value problem.  
A \wtg with non-negative weights has a \emph{strictly non-Zeno cost property} when every cycle is of cost at least $1$. This more than limits Zeno-like behaviours, as strictly non-Zeno \wtgs can be unfolded into acyclic \wtgs; hence, the "Value Problem" is decidable, as shown by Bouyer, Cassez, Fleury and Larsen~\cite{BCFL-fsttcs04}). Busatto-Gaston, Monmege and Reynier~\cite{Busatto2017} generalized this result to \wtgs with arbitrary weights, for which non-Zenoness becomes \emph{divergence}.\footnote{A \wtg is \emph{divergent} if every strongly connected component has either cycles of weight in $\intof{-\infty}{-1}$ or 
cycles of weight in $\intfo{1}{\infty}$.}

\begin{figure}[h]
		\begin{tabular}{rcccc}
		\toprule
		Number of
		&\multirow{2}{*}{Weights in}
		& \multirow{2}{*}{WTG} 
		& {Almost divergent}
		& Divergent \\
		clocks & & & WTG & WTG\\
		\midrule
		\multirow{2}{*}{$1$}&$\Nat$ 
			& Decidable \cite{bouyer2006} 
			& Decidable 
			& Decidable\\
		&$\Rel$
			& Decidable \cite{monmege2024journal}
			& Decidable 
			& Decidable\\ 
		\midrule
		\multirow{2}{*}{$2$}& $\Nat$
			& \textcolor{blue}{\textbf{Undec.} }
			& \textcolor{blue}{\textbf{Decidable} }
			& \multirow{1}{*}{Decidable}\\
		&$\Rel$
			&\textcolor{blue}{\textbf{Undec.} }	
			& \textcolor{blue}{\textbf{Decidable} }	
			& Decidable\\
		\midrule
		\multirow{4}{*}{$\geq 3$}&\multirow{2}{*}{$\Nat$}
			&\multirow{2}{*}{Undec.}	
			&Undec.	\cite{Bouyer2015}
			&\multirow{2}{*}{Decidable \cite{BCFL-fsttcs04}}\\
	&&&\textit{(Approximable)}\cite{Bouyer2015}&\\
	&\multirow{2}{*}{$\Rel$}
		&Undec.	
		&Undec. 
		&\multirow{2}{*}{Decidable \cite{Busatto2017}}\\
	&& \textit{(Non approx.)}\cite{guilmant2024}& \textit{(Approximable)} \cite{Busatto2018}&\\ 
		\end{tabular}

\begin{tikzpicture}[overlay]
\def\x{4}
\def\xx{6.7}
\def\xxx{10.1}
\def\y{4.2}
\def\yy{3.1}
\def\yyy{2}
\def\yyyy{0.2}
\def\z{7}
\def\zz{9.8}
\draw[red,ultra thick,rounded corners] 
(\x,\y) 
-- (\x,\yy)
-- (\xx,\yy)
-- (\xx,\yyy)
-- (\xxx,\yyy)
-- (\xxx,\yyyy);
\draw[red,ultra thick,dotted] (\xxx,\yyyy)-- (\xxx,\yyyy-.5);

\end{tikzpicture}
\caption{Landscape of Value computability and approximability.}
\label{fig-landscape}
\end{figure}

Divergence properties can be weakened into almost divergence properties. A \wtg with non-negative weights is said to be \emph{almost non-Zeno} (or almost strictly non-Zeno, or almost strongly non-Zeno) if its cycles are of weight $0$, or at least $1$. Bouyer, Jaziri and Markey~(\cite{Bouyer2015}) established that the Value of such \wtgs is approximable\footnote{
\ie it can be computed with arbitrary precision.} (but still undecidable). Busatto-Gaston, Monmege and Reynier~(\cite{Busatto2018}) extended this result to \emph{almost divergent}
 \wtgs with arbitrary weights. However, the Approximation problem is undecidable in general, as shown by Guilmant and Ouaknine~\cite{guilmant2024}.

\subsection*{Contributions}

 In this article, we address the remaining open case of this decidability landscape: two-clock weighted timed games. The result we obtained are twofold: On the one hand, we show that the Value problem is undecidable for two-clock \wtgs. This result holds even for \wtgs with non-negative weights under a time-bound restriction. On the other hand, the Value problem is decidable for almost divergent two-clock \wtgs.

This article is the journal version of two conference papers: Our undecidability result presented in \cref{sec-undec} corresponds to \cite{two-clock-undec}.
In \cref{sec-dec}, we show decidability for almost-divergent two-clock \wtgs. While the conference version \cite{vialard2025decidingvaluetwoclocknonzeno} established this result for \wtgs with non-negative weights, the present article extends it to arbitrary weights. This extension is due to Wuttke, based on work conducted under the supervision of Vialard.

\section{Definitions}

Let $\mathcal{X}$ be a finite set of \textbf{clocks}. \textbf{""Clock constraints""} (or \textbf{guards}) over $\mathcal{X}$ are
	expressions of the form $x \mathrel{\bowtie} n$, where $x,y\in\mathcal{X}$ are clocks,
        ${\bowtie}\in\{<,\leq,=,\geq,>\}$ is a comparison symbol, and
        $n\in\mathbb{N}$ is a natural number.  We write $\mathcal{C}$ to denote
        the set of all clock constraints over $\mathcal{X}$.
A \textbf{valuation} on $\mathcal{X}$ is  a function $\nu:\mathcal{X}\to\mathbb{R}_{\geq 0}$. 
	For $d\in\mathbb{R}_{\geq 0}$ we denote by $\nu+d$ the valuation such
        that, for every clock $x \in \mathcal{X}$, $(\nu+d)(x)=\nu(x)+d$. Let
        $X \subseteq \mathcal{X}$ be a subset of all clocks. We write
        $\nu[X :=0]$ for the valuation such that, for every clock $x
        \in X$, $\nu[X :=0](x) = 0$, and $\nu[X:=0](y) = \nu(y)$ for
        all other clocks $y \notin X$.
	For $C \subseteq \mathcal{C}$ a set of clock constraints over $\mathcal{X}$, we say that
        the valuation $\nu$ \textbf{satisfies} $C$, denoted $\nu \models
        C$, if and only if all the comparisons in $C$ hold
        when replacing each clock $x$ by its corresponding value $\nu(x)$.

\begin{definition}
	A \textbf{(turn-based) ""weighted timed game""} is given by a
        tuple $\mathcal{G}=$\linebreak $(L_\mathsf{Min}, L_\mathsf{Max},G,\mathcal{X},T,w)$, where:
	\begin{itemize}
        \item $L_\mathsf{Min}$ and $L_\mathsf{Max}$ are the (disjoint)
          sets of \textbf{locations} belonging to Players $\mathsf{Min}$ and
          $\mathsf{Max}$ respectively; we let $L = L_\mathsf{Min} \cup
          L_\mathsf{Max}$ denote the set of all locations. (In drawings,
          locations belonging to $\mathsf{Min}$ are depicted by blue
          circles, and those belonging to $\mathsf{Max}$ are depicted
          by red squares.)
          \item $G\subseteq L_{\mathsf{Min}}$ are the \textbf{goal locations}.
		\item $\mathcal{X}$ is a set of clocks.
		\item $T\subseteq (L \setminus G) \times2^\mathcal{C}\times 2^\mathcal{X}\times
                  L$ is a set of \textbf{(discrete) transitions}. A
                  transition $\ell \xrightarrow{C,X} \ell'$
                  enables moving from location $\ell$ to
                  location $\ell'$, provided all "clock constraints" in
                  $C$ are satisfied, and afterwards resetting all clocks in $X$
                  to zero. 
		\item $w:(L \setminus G) \cup T\to \mathbb{Z}$ is a \textbf{weight function}.
                \end{itemize}
In the above, we assume that all data (set of locations, set of
clocks, set of transitions, set of "clock constraints") are finite.             
\end{definition}

          Let $\mathcal{G}= (L_\mathsf{Min}, L_\mathsf{Max},G,\mathcal{X},T,w)$ be
          a \wtg.  A \textbf{configuration} over
          $\mathcal{G}$ is a pair $(\ell,\nu)$, where $\ell\in L$ and $\nu$
          is a valuation on $\mathcal{X}$. Let $d \in \mathbb{R}_{\geq 0}$ be a
          \textbf{delay} and $t = \ell \xrightarrow{C,X} \ell' \in T$
          be a discrete transition. One then has a \textbf{delayed
            transition} (or simply a \textbf{transition} if the
          context is clear)
          $(\ell,\nu) \xrightarrow{d,t} (\ell',\nu')$ provided that
          $\nu+d \models C$ and $\nu' = (\nu+d)[X := 0]$. Intuitively,
          control remains in location $\ell$ for $d$ time units, after
          which it transitions to location $\ell'$,
          resetting all the clocks in $X$ to zero in the process.  The
          \textbf{weight} of such a delayed transition is
          $d \cdot w(\ell) + w(t)$, taking into account both of the time
          spent in $\ell$ as well as the weight of the discrete
          transition $t$.

          As noted in~\cite{BusattoGastonMR23},
          without loss of generality one can assume that no
          configuration (other than those associated with goal
          locations) is deadlocked; in other words, for any location
          $\ell \in L\setminus G$ and valuation
          $\nu \in \mathbb{R}_{\geq 0}^{\mathcal{X}}$, there exists
          $d \in \mathbb{R}_{\geq 0}$ and $t \in T$ such that
          $(\ell,\nu) \xrightarrow{d,t} (\ell',\nu')$.\footnote{This
            can be achieved by adding unguarded transitions to a
            sink location for all locations controlled by $\mathsf{Min}$, and
            unguarded transitions to a $\Minsf$'s location with negative-weight unguarded self-loop, from which $\Minsf$ can reach a goal location for the ones
            controlled by $\mathsf{Max}$.}

          Let $k\in\mathbb{N}$.  A \textbf{run} $\rho$ of length $k$ over $\mathcal{G}$ from a given configuration
          $(\ell_0,\nu_0)$ is a sequence of matching delayed
          transitions, as follows:
          \[ \rho = (\ell_0,\nu_0) \xrightarrow{d_0,t_0}
            (\ell_1,\nu_1) \xrightarrow{d_1,t_1} \cdots
            \xrightarrow{d_{k-1},t_{k-1}}  (\ell_k,\nu_k) \, . \]
          The \textbf{weight} of $\rho$ is the
          cumulative weight of the underlying delayed transitions:
          \[ \mathsf{weight}(\rho) = \sum_{i=0}^{k-1} (d_i\cdot
            w(\ell_i) + w(t_i)) \, . \]
          An infinite run $\rho$ is defined in the obvious way;
          however, since no goal location is ever reached, its weight is defined to be infinite:
          $\mathsf{weight}(\rho)=+\infty$.

          A run is \textbf{maximal} if it is either infinite or cannot be extended
          further. Thanks to our deadlock-freedom assumption, finite maximal
          runs must end in a goal location. We refer to maximal runs
          as \textbf{plays}.

          We now define the notion of \textbf{strategy}. Recall
          that locations of $\mathcal{G}$ are partitioned into sets $L_{\mathsf{Min}}$
          and $L_{\mathsf{Max}}$, belonging respectively to Players
          $\mathsf{Min}$ and $\mathsf{Max}$.
          Let Player $\mathsf{P} \in
          \{\mathsf{Min},\mathsf{Max}\}$, and 
write $\mathcal{FR}_{\mathcal{G}}^{\mathsf{P}}$
          to denote the collection of all non-maximal finite runs of $\mathcal{G}$ 
          ending in a location belonging to
          Player $\mathsf{P}$. A \textbf{strategy}
         for Player $\mathsf{P}$ is a mapping
           $\sigma_{\mathsf{P}} : \mathcal{FR}_{\mathcal{G}}^{\mathsf{P}}
           \to \mathbb{R}_{\geq 0} \times T$ such that for all 
           finite runs $\rho \in \mathcal{FR}_{\mathcal{G}}^{\mathsf{P}}$
           ending in configuration $(\ell,\nu)$ with
           $\ell \in L_{\mathsf{P}}$, the delayed transition
           $(\ell,\nu) \xrightarrow{d,t} (\ell',\nu')$ is valid,
           where $\sigma_{\mathsf{P}}(\rho)=(d,t)$ and $(\ell',\nu')$
           is some configuration (uniquely determined by
           $\sigma_{\mathsf{P}}(\rho)$ and $\nu$).

\AP Let us fix a starting configuration $(\ell_0,\nu_0)$, and let
$\sigma_{\mathsf{Min}}$ and $\sigma_{\mathsf{Max}}$ be strategies for
Players $\mathsf{Min}$ and $\mathsf{Max}$ respectively (one speaks of
a \emph{strategy profile}). Let us denote by
$\intro*\play_\mathcal{G}((\ell_0,\nu_0),\sigma_{\mathsf{Min}},\sigma_{\mathsf{Max}})$
the unique maximal run starting from configuration
$(\ell_0,\nu_0)$ and unfolding according to the strategy profile
$(\sigma_{\mathsf{Min}},\sigma_{\mathsf{Max}})$: in other words,
for every strict finite prefix $\rho$ of
$\play_\mathcal{G}((\ell_0,\nu_0),\sigma_{\mathsf{Min}},\sigma_{\mathsf{Max}})$
in $\mathcal{FR}_{\mathcal{G}}^{\mathsf{P}}$, the delayed transition
immediately following $\rho$ in
$\play_\mathcal{G}((\ell_0,\nu_0),\sigma_{\mathsf{Min}},\sigma_{\mathsf{Max}})$
is labelled with $\sigma_{\mathsf{P}}(\rho)$.

Recall that the objective of Player $\mathsf{Min}$ is to reach a goal
location through a play whose weight is as small as possible. Player
$\mathsf{Max}$ has the opposite objective, trying to avoid goal
locations, and, if this is not possible, to maximise the cumulative weight of
any attendant play. This gives rise to the following two symmetrical definitions:
\begin{align*}
\overline{\mathsf{Val}}_\mathcal{G}(\ell_0,\nu_0) &=
  \inf_{\sigma_{\mathsf{Min}}} \left\{ \sup_{\sigma_\mathsf{Max}}
                                              \left\{
                                              \mathsf{weight}(\play_\mathcal{G}((\ell_0,\nu_0),\sigma_{\mathsf{Min}},\sigma_{\mathsf{Max}}))\right\}
  \right\} \mbox{\ and}\\
\underline{\mathsf{Val}}_\mathcal{G}(\ell_0,\nu_0) &=
  \sup_{\sigma_{\mathsf{Max}}} \left\{\inf_{\sigma_\mathsf{Min}}
                                               \left\{\mathsf{weight}(\play_\mathcal{G}((\ell_0,\nu_0),\sigma_{\mathsf{Min}},\sigma_{\mathsf{Max}}))
                                               \right\} \right\} \, .
\end{align*}

$\overline{\mathsf{Val}}_\mathcal{G}(\ell_0,\nu_0)$ represents the smallest 
weight that Player $\mathsf{Min}$ can possibly achieve,
starting from configuration $(\ell_0,\nu_0)$,
against Player $\mathsf{Max}$'s best response, and conversely for
$\underline{\mathsf{Val}}_\mathcal{G}(\ell_0,\nu_0)$: the latter represents
the largest possible weight that Player $\mathsf{Max}$ can enforce,
against best play from Player $\mathsf{Min}$.\footnote{Technically
  speaking, these values may not be literally achievable; however, if the value is finite, then
  given any $\varepsilon > 0$, both players are guaranteed to have
  strategies that can take them to within $\varepsilon$ of the optimal value.}
\AP As noted in~\cite{BusattoGastonMR23}, turn-based
\wtgs are \emph{determined}, and therefore
$\overline{\mathsf{Val}}_\mathcal{G}(\ell_0,\nu_0) =
\underline{\mathsf{Val}}_\mathcal{G}(\ell_0,\nu_0)$ for any starting
configuration $(\ell_0,\nu_0)$; we denote this common value by
$\intro*\Value{\mathcal{G}}(\ell_0,\nu_0)$.

\begin{remark}
	Note that $\Value{\mathcal{G}}(\ell_0,\nu_0)$ can take on real numbers, or either
	of the values $-\infty$ and $+\infty$.
\end{remark}

\section{Undecidability}
\label{sec-undec}
\begin{restatable}{theorem}{thmUndec}\label{thm-undec}
The
Value Problem is undecidable for \wtgs with two clocks or more, even when restricted to non-negative weights and under a time-bound.
\end{restatable}

To establish undecidability, we reduce the Halting Problem for "two-counter machines" to the Value Problem.
A ""two-counter machine"" is a tuple $\mathcal{M}=(Q,q_i,q_h,T)$ 
where $Q$ is a finite set of locations, $q_i,q_h\in Q$ are the initial and final locations and $T\subseteq (Q\times\{c,d\}\times Q)\cup(Q\times\{c,d\}\times Q\times Q)$ is a set of transitions. 
A "two-counter machine" is deterministic if for any location $q$ there is at most one transition $t\in T$ which has $q$ as first component.
As its name suggests, a "two-counter machine" comes equipped with two counters, $c$ and $d$, which are variables with values in $\Nbb$. 
The semantics is as follows: a transition $(q,e,q')$ increases the
value of counter $e\in\{c,d\}$ by $1$ and moves to location $q'$. A
transition $(q,e,q',q'')$ moves to $q'$ if $e=0$ and to $q''$
otherwise. In the latter case, it also decreases the value of $e$ by $1$. 
The Halting Problem for (deterministic) "two-counter machines" is known to be undecidable 
(see \cite[Thm.~14-1]{Minsky-Computations}).

\subsection{Overview of the reduction}

\AP Let $\intro*\Mcal$ be a "two-counter machine" with counters $c$ and $d$. We construct a \wtg $\intro*\GMcal$ using two clocks, $x$ and $y$.
Player $\Minsf$ is responsible for simulating the behavior of $\Mcal$, while $\Maxsf$ is given the opportunity to punish any incorrect simulation. Each punishment by $\Maxsf$ leads directly to the goal location, thereby terminating the game.

Our construction satisfies the following proposition:

\begin{restatable}{proposition}{propReduc}\label{prop:reduc}
Let $\Mcal = (Q, q_i, q_h, T)$ be a deterministic two-counter machine.
\begin{itemize}
\item If $\Mcal$ does not halt, then the \wtg $\GMcal$, starting from the configuration $(q_i, \mathbf{0})$, has value at most $64$.
\item If $\Mcal$ halts in at most $N$ steps, then $\GMcal$, starting from the configuration $(q_i, \mathbf{0})$, has value at least $64 + \frac{11}{12 \times 30^{5N}}$.
\end{itemize}
\end{restatable}

Intuitively, if $\Mcal$ does not halt, then $\Minsf$ can faithfully simulate its infinite execution for an arbitrary number of steps. The longer she plays, the closer the accumulated weight is to $64$ when she eventually exits to a goal location.


On the other hand, if $\Mcal$ halts in $N$ steps, $\Minsf$ has only two options: either simulate the execution faithfully for at most $N$ steps, and exit, yielding a value strictly greater than $64$, or attempt to cheat in order to make the game longer. The key point is that we show that in order to push the computation further, she will eventually cheat by some  quantity bounded from below depending on $N$. Detecting this, $\Maxsf$ is given the opportunity to enforce a punishment yielding a weight of at least $64 + \frac{11}{12 \times 30^{5N}}$.

In the encoding of $\Mcal$ into $\GMcal$, control locations of $\Mcal$ become locations of weight $30$ in $L_\Minsf$. When entering one of these locations, a faithful encoding of the counters $c$ and $d$ is represented by a clock valuation
\[x=1-\frac1{2^{c}3^{d}5^n}, \, y=0 \,,\]
where $n$ denotes the number of simulated steps of $\Mcal$ so far.

Note that when $x = 1 - \frac{1}{2^{c} 3^{d} 5^n}$, reaching a configuration where $x = 1 - \frac{1}{2^{c+1} 3^{d} 5^{n+1}}$ (e.g., when simulating an increment of counter $c$, which also increments $n$) requires waiting for $\frac{9}{10}(1 - x)$ time units.
Similarly, to simulate a decrement of counter $c$, or an increment/decrement of counter $d$, or a simple increment of $n$, one must wait $\alpha \cdot (1 - x)$ time units, where $\alpha$ is one of
$
\frac35,\frac{14}{15},\frac25$ or $\frac45$,
respectively.

Hence, each increasing transition $(q, e, q') \in T$ is simulated using a module of the following structure: $\Minsf$ selects a delay to update $x$, and then $\Maxsf$ is given an opportunity to punish her if the delay is incorrect (see \cref{fig:teCount}			
).
		\begin{center}
			\begin{figure}[h] 
			\centering
							\includegraphics{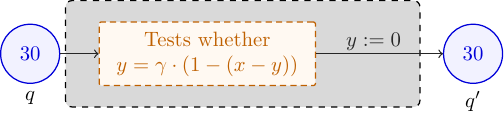}
				\caption{Transition module for increments. 
				Here $x-y$ is the previous value of $x$, and $y$ contains the delay $\Minsf$ waited in $q$.\footnote{In this paper, we follow the conventions below to represent \wtgs:
Blue circles represent locations controlled by $\Minsf$.
Red squares represent locations controlled by $\Maxsf$.
Green circles are goal locations.
Grey rectangles represent modules (i.e., subgames): a transition entering a module transfers control to its starting location. Modules may have outgoing edges.
 Orange rectangles provide the specification of modules.
 Numbers attached to locations denote their respective weights. Numbers in grey boxes attached to arrows denote the weight of the corresponding transition. Transitions without such boxes have weight $0$.
 Green boxes attached to transitions are comments (or assertions) on the values of the clocks, which hold upon taking the transition. These may be complemented by orange boxes, which represent assertions on the corresponding cost incurred.
Some locations are decorated with a numbered flag for ease of reference in proofs.}	
				}\label{fig:teCount}
			\end{figure}
		\end{center}
A branching transition $(q, e, q_z, q_{nz}) \in T$, where $q, q_z, q_{nz} \in Q$ and $e \in \{c, d\}$, is simulated in two steps: first, $\Minsf$ decides whether to move toward $q_z$ or toward $q_{nz}$. In both cases, $\Maxsf$ is given the opportunity to punish if the encoding in $x$ is not valid with respect to $\Minsf$'s choice. Then, $\Minsf$ updates $x$, and $\Maxsf$ is again given the opportunity to punish her for an invalid update.

		\begin{center}
			\begin{figure}[H]
			\centering
				\includegraphics[scale=.85]{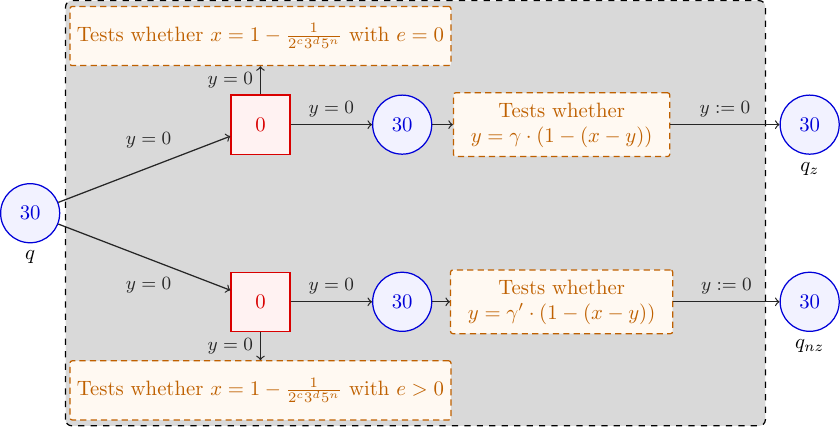}
				\caption{Transition module for branching decrement of counter $e\in\{c,d\}$.} \label{fig:teIf}
			\end{figure}
		\end{center}
		
		Finally, for every $q\in Q$ (including $q_h$), we add the following exit module for  $\Minsf$:
		\begin{center}
			\begin{figure}[H] 
			\centering
				\includegraphics{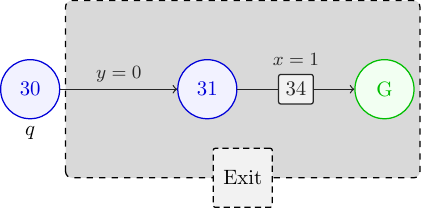}
				\caption{ $\Minsf$'s exit module. Yields a final cost of at most $64+ 1 - x$.}\label{fig:exit}
			\end{figure}
		\end{center}
		
		Note that, while faithfully simulating an infinite execution of $\Mcal$, $x$ will approach $1$ arbitrarily closely, since the encoding $\frac{1}{2^{c} 3^{d} 5^{n}}$ tends to $0$ as the number of steps $n$ increases. Therefore, when $\Minsf$ decides to exit, the accumulated cost will be
$30 x + 31(1-x) +34= 64 + 1-x$.

In \cref{sec-CEC}, we present the punishment module that enforces $y = \alpha(1 - x)$. \cref{sec-IfZero} then introduces punishment modules for cases where $\Minsf$ cheats on branching decisions. Finally, in \cref{sec-combining}, we combine these modules to construct $\GMcal$ and prove \cref{prop:reduc}.
\subsection{Controlling counter evolution}
\label{sec-CEC}

\AP Let us introduce the gadget behind the punishment module ``Punish if $y \neq \alpha(1 - x)$'': the $\CCE$~(Counter Evolution Control) module. The module $\intro*\CCE_{\alpha,\beta}^{M}(x,y)$, depicted in \cref{fig:CCE}, requires that $0 \leq y \leq x < 1$. We will denote the initial values of $x$ and $y$ as $a + b$ and $b$, respectively.

\begin{center}
	\begin{figure}[H]
	\centering
		\includegraphics{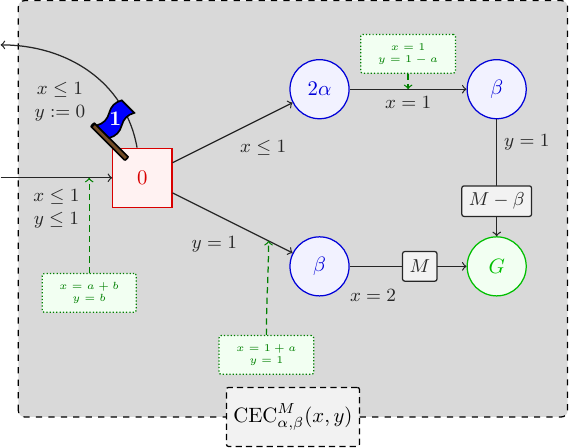}
		\caption{The CEC (Counter Evolution Control) module. Enforces 
$b=\gamma(1-a)$ with $\gamma = 1 - \frac\beta\alpha$.}
		\label{fig:CCE}
	\end{figure}
\end{center}

\begin{proposition}\label{lem:cce}\label{cor:cce1}
	Let $0\leq\beta<\alpha$. Let $t$ be the time $\Maxsf$ spends in location \includegraphics[scale=.5]{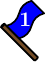} of $\CCE_{\alpha,\beta}^{M}(x,y)$.
	Provided that, upon entering
        the location
        \includegraphics[scale=.5]{standalones/flags/flag1.pdf} in
        $\CCE_{\alpha,\beta}^{M}(x,y)$, $(x,y)=(a+b,b)$ with $0\leq b\leq a+b \leq 1$, and assuming that the overall cost accumulated so far is $\alpha(a+b)+E$, then $\Maxsf$ can either end the game with a final cost of 
	\[
		\alpha\pa{1+\left|b-\pa{1-\f\beta\alpha}(1-a)\right|}+E + M
	\]
	or exit the module with $x=a+b+t$, $y=0$ and accumulated cost $\alpha(a+b+t) + (E-\alpha t)$.
\end{proposition}
\begin{proof}
$\Maxsf$ can either:
	\begin{enumerate}
		\item exit the module after waiting $t$ time units.
		\item take the upper path, which ends the game with cost
		\begin{align*}
		(\alpha-\beta)(1-a)-\alpha b -2\alpha t + \alpha+\beta+M-\beta+E\\
			= (\beta-\alpha)a - \alpha b - 2\alpha t +
                        2\alpha + M - \beta+E \, .\end{align*}
		\item take the lower path, which ends the game with cost
		\[\alpha b - (\alpha-\beta)(1-a) + \alpha+M+E = \alpha
                  b - (\beta-\alpha)a + \beta + M +E \, . \qedhere
                  \]
	\end{enumerate}
\end{proof}

Thus, if we take $\alpha=30$ and $\beta\in \left\{18, 12, 6, 3,
  2\right\}$, the value of $b$ minimising the cost of $\Maxsf$
reaching the goal location is indeed \[
b= \gamma(1-a) \text{ for }
	\gamma\in\left\{\f25, \f35, \f45, \f9{10}, \f{14}{15}\right\}\;.
\]

For $\alpha=30$, $E\leq 0$ and $M=34$, note that when $b= \gamma(1-a)$,
$\Maxsf$'s punishment yields cost at most $64$, thus $\Maxsf$ has no incentive to punish $\Minsf$ when she has not cheated.

\subsection{Controlling whether a counter is zero}
\label{sec-IfZero}

We now address the case in which $\Minsf$ has moved to the wrong location when simulating a zero-test. To handle this, depending on the situation, $\Maxsf$ has access to either a control-if-zero ($\CZ$) or control-if-not-zero ($\CNZ$) module, which checks whether a counter is indeed zero or non-zero by forcing a series of multiplications that can reach 1 only if $\Minsf$ chose the right branch. We begin by presenting the module that controls the multiplication ($\intro*\CM$), depicted in \cref{fig:CM}.

\begin{center}
	\begin{figure}[H]
	\centering
		\includegraphics{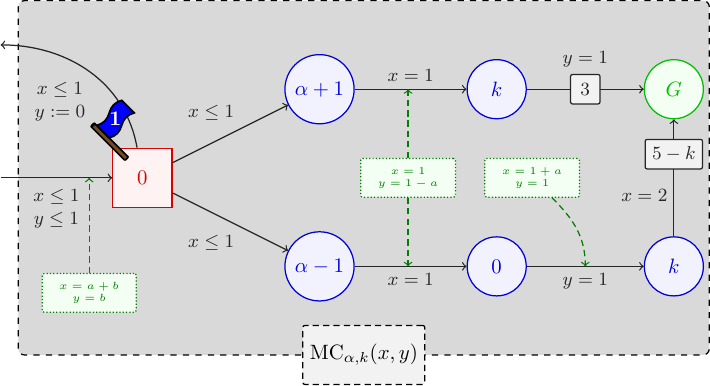}
		
		\caption{The $\CM$~(Multiplication Control) module. Enforces 
$x=ka$ for $k\in\{2,3,5\}$.}
		\label{fig:CM}
	\end{figure}
\end{center}

$\CM_{\alpha,k}(x,y)$ operates similarly to the $\CCE$~module:
\begin{itemize}
\item It requires that $0 \leq y \leq x < 1$. We will denote the initial values of $x$ and $y$ as $a + b$ and $b$, respectively.
\item Whereas the $\CCE$~module enforces that $b = \gamma(1 - a)$ for some $\gamma$, $\CM_{\alpha,k}(x,y)$ ensures that $a + b = k a$. This is because the $\CM$~module is designed to multiply encodings of the form $\frac{1}{2^{c} 3^{d} 5^{n}}$ by $k \in {2,3,5}$.
\end{itemize}


\begin{proposition}\label{lem:xm}\label{cor:cm1}
Let $1\leq\alpha$.
	Let $t$ be the time $\Maxsf$ spends in location \includegraphics[scale=.5]{standalones/flags/flag1.pdf} of $\CM_{\alpha,k}(x,y)$.
	Provided that, upon entering
        the location
        \includegraphics[scale=.5]{standalones/flags/flag1.pdf} in
        $\CM_{\alpha,k}(x,y)$, $(x,y)=(a+b,b)$ with $0\leq b\leq a+b \leq 1$, and assuming that the overall cost accumulated so far is $\alpha(a+b)+E$, 
	$\Maxsf$ can either end the game with final cost
	\[
	\alpha+4+E+\left|b- (k-1)a\right|
	\]
	or exit the module with $x=a+b+t$, $y=0$ and an accumulated cost of $\alpha(a+b+t) + (E-\alpha t)$.
\end{proposition}
\begin{proof}
$\Maxsf$ can either:
	\begin{enumerate}
		\item exit the module after waiting $t$ time units.
		\item take the upper path, which ends the game with cost
		$$\alpha(a+b)+E+(\alpha+1)(1-a-b-t)+ka+3
		=
                  \alpha + 4+ E -(\alpha+ 1)t-b + (k-1)a\,.$$
		\item take the lower path, which ends the game with cost 
		\begin{align*}
		\alpha(a+b)+E+(\alpha-1)(1-a-b-t)+k(1-a)+5 - k\\
		=
		\alpha + 4 + E - (\alpha-1)t + b  - (k-1)a\,.  \qquad \qedhere
		\end{align*}
	\end{enumerate}
\end{proof}

Here we see clearly that the $\CM$~module enforces multiplication: the cost of $\Maxsf$ ending the game is minimised for  $b=(k-1)a$, hence $x=a+b= ka$.

\AP We now introduce modules to control whether a counter is indeed $0$ or not.
\begin{center}
	\begin{figure} 
	\centering
		\includegraphics[scale=.9]{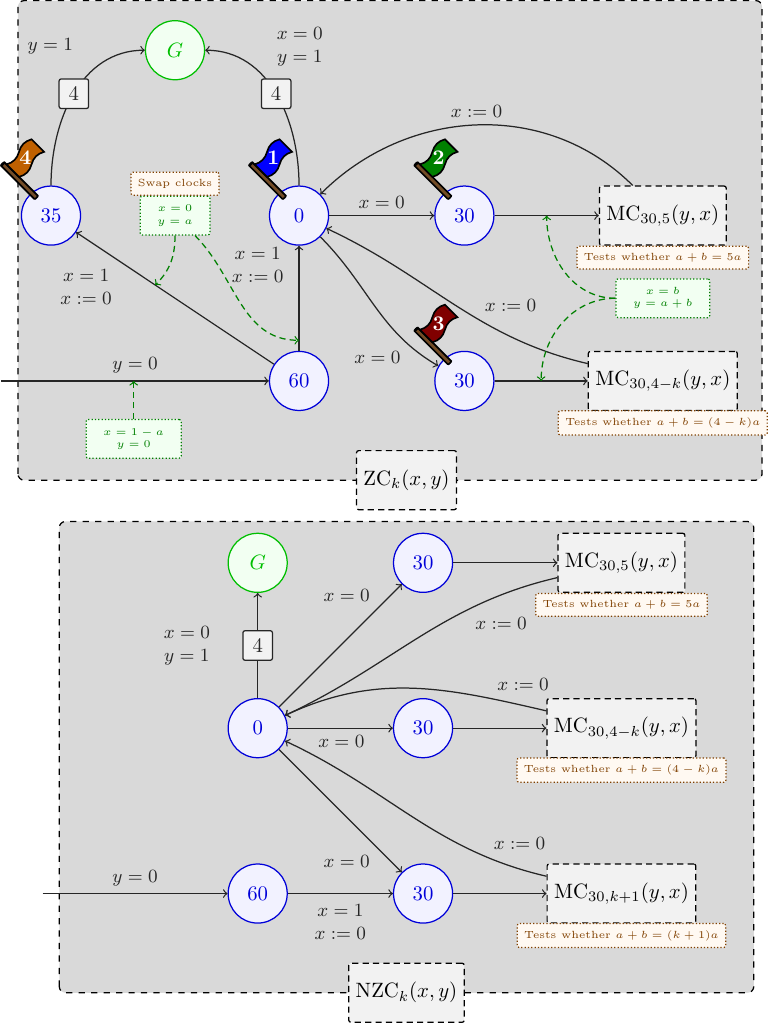}
		\caption{The Zero-Control ($\intro*\CZ$) and Non-Zero-Control ($\intro*\CNZ$) modules, for $k\in\{1,2\}$.}\label{fig:cz}
	\end{figure}
\end{center}

In these two modules, note that $\CM$~is always invoked as $\CM(y,x)$, i.e., by swapping the roles of the two clocks.
This is because it is easier to translate the encoding $(1-a,0)$ into
$(0,a)$ rather than $(a,0)$. This translation is the first step in both modules.

\begin{restatable}{proposition}{propcz}\label{prop:cz}
	Let $k\in\{1,2\}$.
	Let $1-a\in\intfo01$ be the initial value of $x$ upon entering the module $\CZ_k(x,y)$. 
	Let $30(1-a)+E$ be the overall cost accumulated to date upon entering the module.
	Let 
	\[
		\mu=\min\left\{\left|\f1{(4-k)^{d}5^{n}}-a \right| \colon
                  d,n\in\Nbb\right\} \, .
	\]
	Then 
	\begin{itemize}
		\item $\Minsf$ has a strategy that ensures a final cost of at most $64+E+5\mu$.
		
		\item $\Maxsf$ has a strategy that ensures a final cost of at least $64+E+\mu$.
	\end{itemize}
\end{restatable}

\begin{proof}
	Let $p\in\Nbb $. We introduce the following quantities:
	\begin{itemize} 
		\item $a_p$ is the value of clock $y$ upon entering
                  \includegraphics[scale=.5]{standalones/flags/flag1.pdf}
                  for the $p$-th time (if it exists).
		
		\item $C_p$ is the accumulated cost so far upon
                  entering
                  \includegraphics[scale=.5]{standalones/flags/flag1.pdf}
                  for the $p$-th time (if it exists).
		
		\item $t_p$ is the time spent by $\Minsf$ in one of the locations \includegraphics[scale=.5]{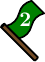} or \includegraphics[scale=.5]{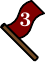} upon leaving
		\includegraphics[scale=.5]{standalones/flags/flag1.pdf}
                for the $p$-th time (if it exists).
		
		\item $\epsilon_p$ is the time $\Maxsf$ waits in one
                  of the locations of $\CM$~controlled by him for the $p$-th time (if it exists).
		
		\item $E_p=C_p-30(1+a_p)$ is the difference between the actual and
                  expected costs.
                  \footnote{$30(1+a_p)$ is the accumulated cost in a run where $\Maxsf$ has never waited a delay $>0$.}
	\end{itemize}
	We have the following initial conditions:
	\[
		a_0=a 
		\qquad C_0=30(1+a) + E_0
		\qqandqq E_0=E\, .
	\]
	When everything is well defined we have the following recurrence relations:
	\[
		a_{p+1}=a_p+t_p+\epsilon_p
		\qquad C_{p+1}=C_p+30t_p
		\qqandqq E_{p+1} = E_p - 30\epsilon_p \, .
	\]
	Overall we have
	\[
		a_p = a + \Sum{q=0}{p-1}(t_q+\epsilon_q)
		\qquad E_p=E-30\pa{\Sum{q=0}{p-1}\epsilon_q}
		\qqandqq C_p= 30(1+a_p) + E_p  \, .
	\]

	\begin{itemize}
		\item \textbf{Upper bound:} 
First, let us give a $\Minsf$ strategy which ensures a final cost of at most $64 + E + 5\mu$. 		Let $d_p$ and $n_p$ be integers that minimise $\mu_p\overset{\text{def}}{=} |\delta_p|$ for 
		$\delta_p\overset{\text{def}}{=}
                \f1{(4-k)^{d_p}5^{n_p}}-a_p$. In other words, $\mu_p$ is the distance between $a_p$ and the closest valid encoding for the 
$\CZ$ module.
\begin{enumerate}
\item First, if $a$ is closer to $1$ than to $\f{1}{4-k}$ (the second largest valid encoding), then $\Minsf$ takes the transition to location~\includegraphics[scale=.5]{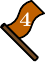}; otherwise, she moves to location~\includegraphics[scale=.5]{standalones/flags/flag1.pdf}.

\item Assuming that $\Minsf$ is in location~\includegraphics[scale=.5]{standalones/flags/flag1.pdf} for the $p$-th time:
			\item\label{it:c2} If $n_p\geq1$ then go to location~\includegraphics[scale=.5]{standalones/flags/flag2.pdf} and wait $t_p=4a_p + 5\delta_p$. Doing so, $a_p+t_p=\f1{(4-k)^{d_p}5^{n_p-1}}$.
						Note that $t_p\geq0$, since otherwise
\[
\f1{(4-k)^{d_p}5^{n_p}}<\f1{(4-k)^{d_p}5^{n_p-1}}<a_p
\]
			and thus 
			\[
				\left|\f1{(4-k)^{d_p}5^{n_p-1}}-a_p\right|<\mu_p
			\]
			which is a contradiction.

\item
\label{it:c1} Similarly, if $n_p=0$ and $d_p\geq 1$ then go to location \includegraphics[scale=.5]{standalones/flags/flag3.pdf} and wait $t_p=(3-k)a_p + (4-k)\delta_p$. Doing so, $a_p+t_p=\f1{(4-k)^{d_p-1}}$.
						Note that for the same reason as earlier, $t_p\geq0$.
						
\item						
Finally, if the closest valid encoding is $1$ ($d_p=n_p=0$) but $a_p < 1$, wait $t_p=1-a_p$ in location~\includegraphics[scale=.5]{standalones/flags/flag3.pdf}.

\item When $a_p = 1$ in location~\includegraphics[scale=.5]{standalones/flags/flag1.pdf}, exit to the goal location.
\end{enumerate}

Under this strategy, $\Minsf$ always selects $t_p$ such that $t_p+a_p$ is a valid encoding.
Hence, any ``error'' in the $\CM$~module arises either from the initial deviation $\mu$ or from a delay 
	\[\epsilon_{p-1}\in\left\{\left|\f1{(4-k)^{d}5^{n}}-a_p\right|
            \colon d,n\in\Nbb\right\} \, .\]
introduced by $\Maxsf$ in the previous step, thus $\mu_p\leq \epsilon_{p-1}$ for $p>0$.

	\begin{itemize}
		\item $\Minsf$ goes directly to location~\includegraphics[scale=.5]{standalones/flags/flag4.pdf}. In this case, we have $a\geq 
                \f12\pa{1+\f1{4-k}}$, which entails that $\mu=1-a$. Hence the final cost is
		\[
			30(1-a) + E + 60a + 35(1-a) + 4 = 64 +
                        E + 5(1-a)=  64 +
                        E + 5\mu\, .
		\]
Similarly, if $\Minsf$ goes to location~\includegraphics[scale=.5]{standalones/flags/flag4.pdf} when $a_p=1$,  then the final cost is exactly
		\[
		64+E_p=64+E-30\Sum{q=0}{p-1}\epsilon_q\leq
                64+E+5\mu \, .
		\]
	
			\item If $\Maxsf$ decides to end the game in one of the \CM~ modules (we put aside for now the case when $t_p=1-a_p$), then if $p=1$, using \cref{cor:cm1}, with optimal play from $\Maxsf$, the final cost is at most
		$
		64+E+5\mu
		$.
		
		If $p>1$, using \cref{cor:cm1}, the final cost is at most
		\begin{align*}
			64+E_p+5\mu_p 
			&= 64 + E - 30\Sum{q=0}{p-1}\epsilon_q + 5\mu_p\\
			&\leq 64 + E - 30\Sum{q=0}{p-2}\epsilon_q - 25\epsilon_{p-1}
			\text{ since $\mu_p\leq \epsilon_{p-1}$}\\
			&\leq 64+E\leq 64 + E +5\mu \, .
		\end{align*}
		Intuitively, if $\Maxsf$ punishes an error that he has introduced, the corresponding punishment cost is offset by the ``lost'' cost incurred by $\Maxsf$ while waiting in a zero-weight location in the previous step.

		\item The last punishment opportunity for $\Maxsf$ is when $\Minsf$ picks $t_p=1-a_p$. Now, observe that in this case $p>1$.
		Indeed, $d_p = n_p=0$ entails that  $a_p$ is closer to $1$ than to $\displaystyle{\frac{1}{4-k}}$,
                which (for $p=1$) mandates moving to
                location~\includegraphics[scale=.5]{standalones/flags/flag4.pdf}
                instead of location~\includegraphics[scale=.5]{standalones/flags/flag1.pdf}.

   We can also note that this situation can only happen if $\Maxsf$ waited
   at least $\f{3-k}{8-2k}$ time units in the previous round. Indeed, following $\Minsf$'s strategy, we know that $a_{p-1}+t_{p-1}$ is of the form $\f1{(4-k)^{d}5^{n}}$ with either $d>0$ or $n>0$.
		Therefore $a_{p-1}+t_{p-1}\leq \f1{4-k}$ and 
		$a_p\geq \f{5-k}{8-2k}$, hence $\epsilon_{p-1}\geq 
		\f{3-k}{8-2k}$.
		
		This means that$
		|1-(4-k)a_p|\leq 3-k \leq 30\epsilon_{p-1}
		$. In other words, the weight lost by $\Maxsf$ by waiting $\epsilon_{p-1}$ more than compensates the punishment he can now apply.

	 Therefore, using \cref{cor:cm1}, the final cost is at most
		\begin{align*}
		64+E_p+|1-(4-k)&a_p|
			= 64 + E - 30\Sum{q=0}{p-1}\epsilon_q +
                        |1-(4-k)a_p|\\
            &\leq 64 + E - 30\Sum{q=0}{p-2}\epsilon_q 
			\leq 64 + E \leq
			64 + E +5\mu \, .
		\end{align*}

	\end{itemize}
	In all the above cases, $\Minsf$ ensures a cost of at most $64 + E +5\mu$.

	
	\item\textbf{Lower bound:}   Let us now give a $\Maxsf$ strategy which ensures a final cost of at least $64 + E + \mu$.

	Assume that $\Minsf$ has gone to location~\includegraphics[scale=.5]{standalones/flags/flag1.pdf} instead of \includegraphics[scale=.5]{standalones/flags/flag4.pdf}.
Let us write:
	\begin{itemize}
		\item $k_p=\begin{accolade}
			5& \text{if $\Minsf$ chooses
                          location~\includegraphics[scale=.5]{standalones/flags/flag2.pdf}
                        at step $p$}\\
			4-k& \text{if $\Minsf$ chooses
                          location~\includegraphics[scale=.5]{standalones/flags/flag3.pdf}
                        at step $p$}
                        \, .
		\end{accolade}$
		\item $\eta_p= \left|t_p-(k_p-1)a_p\right|$.
	\end{itemize}
	Intuitively, $\eta_p$ corresponds to the error $\Minsf$ made at step $p$, and that $\Maxsf$ must now decide whether to punish.
	Consider the following strategy for $\Maxsf$: 
	\begin{enumerate}
		\item Always choose $\epsilon_p=0$.
		\item\label{it:c2max} If $\Minsf$ picks $t_p$ such that
		$\eta_p\geq\mu$, then punish her in the next $\CM$~module through the path maximising the cost.
		\item Otherwise, immediately accept and go back to location \includegraphics[scale=.5]{standalones/flags/flag1.pdf}.
	\end{enumerate}
	
		First observe that if $\Minsf$ goes to \includegraphics[scale=.5]{standalones/flags/flag4.pdf}, 
	then this yields a weight $30(1-a)+E+ 60a +35(1-a)+4= 64 +  E+5(1-a)\geq 64 +  E+5\mu$.

			Otherwise, using \cref{cor:cm1}, if $\Maxsf$ ends the game in an \CM~ module, he secures a final cost of 
	\[
		64+E_p+\eta_p=64+E+\eta_p\geq 64+E+\mu \, .\quad \qedhere
	\]

Let us now prove that $\Minsf$ will need to make a mistake $\eta_p\geq \mu$ in order to reach the goal location:



	For the sake of contradiction, assume that we always have~$\eta_p<\mu$, i.e., $\Maxsf$ never ends the game by himself, and that nonetheless $\Minsf$ can exit the game at step $P$ for
        some $P\in\Nbb$ (hence $a_P=1$). 
	Since $\epsilon_p=0$ for all $p$, and since $\Maxsf$ always
        takes the transition back to
        location~\includegraphics[scale=.5]{standalones/flags/flag1.pdf},
        we have $t_p=a_{p+1}-a_p$ for all $p$. Therefore $\eta_p=|a_{p+1}-k_pa_p|$.

	Let us consider the following property for $p\in\{1,\dots,P\}$:
	\[
		(\Pscr_p)\mathrel:\qquad
                \left|\Prod{q=1}{p}\f1{k_{P-q}}-a_{P-p}\right|<\mu\Sum{m=1}{p}\quad\Prod{q=m}{p}\f1{k_{P-q}}
                \, .
	\]
	In other words, the distance between $a_{P-p}$ and the valid encoding it simulates is bounded.
	We prove this property by induction on $p$.
	\begin{itemize}
		\item We have
                  $|1-k_{P-1}a_{P-1}|=|a_P-k_{P-1}a_{P-1}|=\eta_{P-1}<\mu$. This
                  can be rewritten  as
		\[
			\left|\f1{k_{P-1}}-a_{P-1}\right|<\f\mu{k_{P-1}}
                        \, ,
		\]
		which is exactly $(\Pscr_1)$.
		
		\item Assume that $(\Pscr_p)$ has been proven for some $p\in\{1,\dots,P-2\}$. We have
		\begin{align*}
			&\left|\Prod{q=1}{p}\f1{k_{P-q}} - k_{P-(p+1)} a_{P-(p+1)}\right| \\
				&\quad\leq \left|\Prod{q=1}{p}\f1{k_{P-q}} - a_{P-p}\right|+\left|a_{P-p}-k_{P-(p+1)}a_{P-(p+1)}\right|\\
				&\quad\leq \left|\Prod{q=1}{p}\f1{k_{P-q}} - a_{P-p}\right| + \eta_{P-p}\\
				&\quad<
                           \mu\pa{\Sum{m=1}{p}\quad\Prod{q=m}{p}\f1{k_{P-q}}}
                           + \mu \, .
		\end{align*}
		Thus
		\begin{align*}
			\left|\Prod{q=1}{p+1}\f1{k_{P-q}} - a_{P-(p+1)}\right|
				&< \mu\pa{\Sum{m=1}{p}\quad\Prod{q=m}{p+1}\f1{k_{P-q}}} + \f\mu{k_{P-(p+1)}}\\
				&=
                                \mu\Sum{m=1}{p+1}\quad\Prod{q=m}{p+1}\f1{k_{P-q}}
                                \, .
		\end{align*}
		This proves $(\Pscr_{p+1})$, which concludes our
                induction. 
	\end{itemize}

        We have
	\begin{align*}
		\left|\Prod{q=1}{P}\f1{k_{P-q}}-a_0\right|
		<\mu\Sum{m=1}{P}\quad\Prod{q=m}{P}\underbrace{\f1{k_{P-q}}}_{\leq\f1{4-k}}
			&\leq \mu\Sum{m=1}{P}\f1{(4-k)^{P-m}}\\
			&=
                        \mu\underbrace{\Sum{m=1}{P-1}\f1{(4-k)^{m}}}_{<\f1{3-k}\leq 1}<\mu
                        \, .
	\end{align*}
	But we notice that 
	\[
	\left|\Prod{q=1}{P}\f1{k_{P-q}}-a_0\right| = \left|\Prod{q=1}{P}\f1{k_{P-q}}-a\right|
		\in \left\{\left|\f1{(4-k)^{d}5^{n}}-a\right| \colon
                  d,n\in\Nbb\right\} \, .
	\]
	Thus by definition of $\mu$, 
	\[
		\mu \leq \left|\Prod{q=1}{P}\f1{k_{P-q}}-a_0\right|
                <\mu \, , 
	\]
a contradiction. 


	\end{itemize}
\end{proof}

\begin{proposition}\label{prop:cnz}
	Let $k\in\{1,2\}$.
	Let $1-a\in\intfo01$ be the initial value of $x$ upon entering module $\CNZ_k(x,y)$. 
	Let $30(1-a)+E$ be the overall cost accumulated thus far when entering the module.
	Let 
	\[
	\mu=\min\enstqcolon{\left|\f1{(k+1)^{c}(4-k)^{d}5^{n}}-a\right|}{c,d,n\in\Nbb\quad
          c>0} \, .
	\]
	Then 
	\begin{itemize}
		\item $\Minsf$ has a strategy that ensures a final cost of at most $64+E+5\mu$.
		
		\item $\Maxsf$ has a strategy that ensures a final cost of at least $64+E+\mu$.
	\end{itemize}
\end{proposition}

\begin{proof}
	The proof is almost identical to that of \cref{prop:cz}.
	The main difference between modules $\CZ_k$ and $\CNZ_k$ 
	is that $\CNZ_k$ forces $\Minsf$ to do a multiplication by $k+1$ before multiplying by $k+1$, $4-k$, and $5$ arbitrarily many times.
\end{proof}

\subsection{Combining modules}
\label{sec-combining}
%
%
%
%
%

We can now explicitly implement the modules given in \cref{fig:teCount,fig:teIf} with the 
$\CCE,\CZ$ and $\CNZ$ modules, and prove \cref{prop:reduc}:
		\begin{center}
			\begin{figure}[H] 
			\centering
				\includegraphics{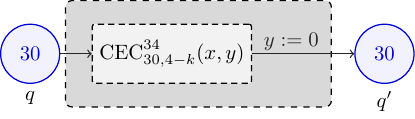}
				\caption{Transition module simulating an increment transition $(q,e,q')\in T$, with $q,q'\in Q$. Here $k=1$ or $2$ when $e=c$ or $d$, respectively.}

			\end{figure}
		\end{center}
		\begin{center}
			\begin{figure}[H]
			\centering
				\includegraphics[scale=.9]{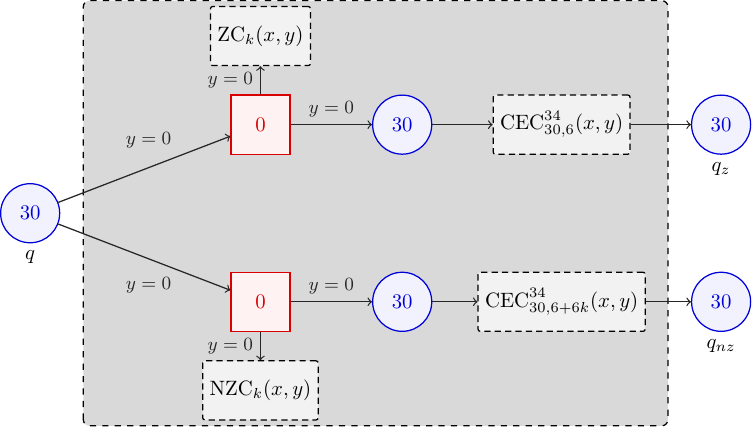}
				\caption{Transition module simulating a branching transition
				$(q,e,q_z,q_{nz})\in T$, with $q,q_z,q_{nz}\in Q$. Here $k=1$ or $2$ when $e=c$ or $d$, respectively.}

			\end{figure}
		\end{center}
\propReduc*
\begin{proof}

	We first introduce some key quantities.
	Let $p\in\Nbb $. 
        We let $\left(q_p,c_p,d_p\right)_p$ be the unique sequence of
        locations and values of the two counters along the
        (deterministic) execution of $\Mcal$. Here, for
        all $p$, $c_p+d_p\leq p$.  Recall that $\Lcal_\Minsf$ is a superset of $Q$.
We let $\left(\hat q_p\right)$ be the
        sequence of locations in $Q$ visited along a given play of
        $\GMcal$. Let us write, upon entering a location $\hat q_p$:
	\begin{itemize} 
		\item $a_p$ to denote the value of clock $x$.
		
		\item $C_p$ for the accumulated cost so far.
		
		\item $E_p=C_p-30a_p$ for the difference between the actual and expected costs.
		
		\item
                  $\mu_p=\min\left\{\left|1-\f1{2^{c}3^{d}5^p}-a_p\right|
                    \colon c,d\in\Nbb \quad
                    c+d\leq p\right\}$ to denote the difference between $a_p$ and the closest valid counter encoding.
		Let $\hat c_p$ and $\hat d_p$ be the counter values of this closest encoding, i.e., $\hat c_p+\hat d_p\leq p$ and
		\[
		\left|1-\f1{2^{\hat c_p}3^{\hat
                      d_p}5^p}-a_p\right|=\mu_p \, .
		\]
		\item Let
		$
		\delta_p = 1-\f1{2^{\hat c_p}3^{\hat d_p}5^p}-a_p 
		$.
	\end{itemize}
	At the start of the game,
	$
	a_0=0, \,\mu_0=\delta_0=0,\, E_0=0$ and  $C_0=0 \, .
	$
	After entering a location $\hat q_p$, and until visiting the next location $\hat q_{p+1}$, $\Minsf$ can wait in a location of weight $30$
	\footnote{This location is either $\hat q_p$ itself, or the next visited location in $\Lcal_\Minsf$  when simulating a zero-test.}
	, immediately followed by a \CCE~module where $\Maxsf$ can also wait.
	We let
$t_{p}$ be the time spent by $\Minsf$ before the $\CCE$~module is entered, and
$\epsilon_{p}$ be the time spent by $\Maxsf$ in his location of the $\CCE$~module before $\hat q_{p+1}$.

	This gives us the following recurrence relations:
	\[
	a_{p+1}=a_p+t_{p}+\epsilon_{p}\qquad C_{p+1}=C_p+30t_{p}
	\qqandqq E_{p+1} = E_p - 30\epsilon_{p} \, .
	\]
	Overall we have
	\[
	a_p = \Sum{q=0}{p-1}(t_q+\epsilon_q)\qquad E_p=-30\Sum{q=0}{p-1}\epsilon_q
	\qqandqq C_p= 30a_p + E_p=30\Sum{q=0}{p-1}t_q \, .
	\]
	\begin{itemize}
		\item \textbf{$\Minsf$'s strategy when $\Mcal$ does not halt:} For any
                  $N\in\Nbb$, consider the following strategy for
                  $\Minsf$: upon entering $q_p$, if $1-\f1{30^N}<a_p$ or $1-\f1{2^{c_{p+1}}3^{d_{p+1}}5^{p+1}}< a_p$ then take the exit module.
Otherwise, $\Minsf$ plays in order to reach $q_{p+1}$, and thus she waits
			\[
			t_{p}=1-\f1{2^{c_{p+1}}3^{d_{p+1}}5^{p+1}}-a_p = \delta_p +\f1{2^{c_p}3^{d_p}5^p}-\f1{2^{c_{p+1}}3^{d_{p+1}}5^{p+1}}\;.
			\]
			In other words, $\Minsf$'s strategy is to simulate the execution of $\Mcal$ faithfully, until either $a_p$ is $\f1{30^N}$-close to $1$, in which case the total cost is
			\[
				64+E_p+1-a_p \leq 64+E_p+
                                 \f1{30^N} \, ,
			\] 
			or $\Maxsf$ has waited long enough in $\epsilon_{p-1}$ to prevent $\Minsf$ from reaching the next step.
			In this case, the weight lost by waiting $\epsilon_{p-1}$ is enough to compensate the cost of leaving for $\Minsf$: 
			\begin{align*}\epsilon_{p-1} = a_p-a_{p-1}-t_{p-1} &= a_p -\left(1 - \f1{2^{c_p}3^{d_p}5^p}\right)\\
				&>\f1{2^{c_p}3^{d_p}5^p}-\f1{2^{c_{p+1}}3^{d_{p+1}}5^{p+1}}
                           \, ,\end{align*}
        hence  $30\epsilon_{p-1}  >\f1{2^{c_{p+1}}3^{d_{p+1}}5^{p+1}} >  1-a_p              $.
            The total cost of $\Minsf$ leaving is then               
                           			\[
				64+E_p+1-a_p <
                                64-30\Sum{q=0}{p-2}\epsilon_q\leq 64
                                 \, .
			\]

		
		Note that, along the game, $\Maxsf$ gains nothing by punishing $\Minsf$.
							Indeed, for any $p>0$ such that the run is defined up to $p$ steps,
		\[
		a_{p}=a_{p-1}+t_{p}+\epsilon_{p}\qqandqq
		a_{p-1}+t_{p-1}=1-\f1{2^{c_{p}}3^{d_{p}}5^{p}} \, .
		\]
		Therefore 
		\[
		\epsilon_{p-1}\in\enstqcolon{\left|1-\f1{2^{c}3^{d}5^{p}}-a_{p}\right|}{c,d\in\Nbb\quad
                  c+d\leq p-1} \, ,
		\]
		hence $\epsilon_{p-1}\geq\mu_{p}$.
		This means that any ``error'' on the encoding at step $p$ must have been introduced by $\Maxsf$ at step $p-1$.
		Thus, according to \cref{prop:cz,prop:cnz,cor:cce1}, the final weight is at most $64<64+\f1{30^N}$.

Overall, for any $N\in\Nbb$, $\Minsf$ can secure a cost of at most $64+\f1{30^N}$. Thus the value of the game $\GMcal$ is at most $64$.

		\item \textbf{Consider the following $\Maxsf$ strategy when $\Mcal$ halts in $N$ steps:}
		\begin{enumerate}
			\item Always choose $\epsilon_p=0$ (thus $E_p=0$ for all $p$).
			\item Immediately punish a wrong transition
                          going to $\CZ_k$ or $\CNZ_k$ if $\Minsf$
                          unfaithfully simulates a zero-test.
			\item At step $p$,
			accept the transition if $\left|t_{p}-\pa{1-\f\beta{30}}(1-a_p)\right|<\f1{30^{5N+1}}$ and punish otherwise.
		\end{enumerate}

First observe that if $\Minsf$ causes		
	$\Maxsf$ to end the game in a $\CCE_{30,\beta}^{34}$ module (Case~$3$), then \cref{cor:cce1} ensures that the final cost is
			\[
			64+E_p+30\left|t_{p}-\pa{1-\f\beta{30}}(1-a_p)\right|
                        \geq 64+\frac1{30^{5N}}\geq64+\frac{11}{12\times 30^{5N}}\,,
			\]
			since $\left|t_{p}-\pa{1-\f\beta{30}}(1-a_p)\right|\geq\f1{30^{5N+1}}$.
		
		However, avoiding this punishment limits $\Minsf$'s possibilities:

\bigskip
\textbf{Claim.} 
		For any $p\leq N$,
                if the game runs at least $p$ steps, then the following property holds:
		\[
			(\Pscr_p)\mathrel: (\hat c_p, \hat
                        d_p)=(c_p,d_p)\wedge \mu_p\leq\f1{12\times
                          30^{5N}} \, .
		\]

\begin{proof}
[Proof of Claim]

We prove $(\Pscr_p)$ by induction on $p$.
\begin{itemize}
			\item $(\Pscr_1)$ holds by definition.
			\item Assume that $(\Pscr_p)$ holds and that
                          location $q_{p+1}$ exists. Reaching $q_{p+1}$ entails that $\Minsf$ chose the correct transition in case of zero-tests and that she did not exit via the exit module. The run then went through a $\CCE_{30,\beta}^{34}$ module with $\beta\in \left\{18, 12, 6, 3,
  2\right\}$ such that
			\[
				\f\beta{30}\f1{2^{c_p}3^{d_p}5^p}=\f1{2^{c_{p+1}}3^{d_{p+1}}5^{p+1}}
                                \, .
			\]
			According to $\Maxsf$'s strategy, we have
			\[
				\left|t_{p}-\pa{1-\f\beta{30}}(1-a_p)\right|<\f1{30^{5N+1}}
                                \, ,
			\]
			since otherwise $\Maxsf$ would have ended the game. Thus
			\begin{align*}
				\left|1-\f1{2^{c_{p+1}}3^{d_{p+1}}5^{p+1}}-a_{p+1}\right| 
					&= \left|1-\f\beta{30}\f1{2^{c_p}3^{d_p}5^p}-a_p-t_{p}\right|\\
					&= \left|1-a_p - \f\beta{30}\pa{1-a_p-\delta_p} - t_{p}\right|\\
					&= \left|\f\beta{30}\delta_p +\pa{1-\f\beta{30}}(1-a_p) - t_{p}\right|\\
					&\leq \f\beta{30}\mu_p + \left|\pa{1-\f\beta{30}}(1-a_p) - t_{p}\right|\\
					&< \f35\mu_p + \f1{30^{5N+1}}\\
					&\leq \f35\f1{12\times
                                   30^{5N}} + \f1{30^{5N+1}} =
                                   \f1{12\times 30^{5N}}  \, .
			\end{align*}
Now we have
			\begin{align*}
				\min&\enstqcolon{\left| \f1{2^{c}3^{d}5^{p+1}} - \f1{2^{c_{p+1}}3^{d_{p+1}}5^{p+1}} \right|}{\begin{array}c
						c,d\in\Nbb\\
						c+d\leq p+1\\
						(c,d)\neq(c_{p+1},d_{p+1})
				\end{array}}\\
				&\geq\f1{30^{p+1}}		
			\geq\f1{30^{N}} 
			\end{align*}					
since any element of the above set is a multiple of $\f1{30^{p+1}}$ and cannot be $0$.

			Therefore, if $(\hat c_{p+1},\hat d_{p+1})\neq (c_{p+1}, d_{p+1}) $,
                        we would have
			\begin{align*}
				&\left|1-\f1{2^{\hat c_{p+1}}3^{\hat d_{p+1}}5^{p+1}}-a_{p+1}\right|\\
				&\quad\geq
				\left|\f1{2^{\hat c_{p+1}}3^{\hat d_{p+1}}5^{p+1}}-
				\f1{2^{c_{p+1}}3^{ d_{p+1}}5^{p+1}}\right|
				- \left|1-\f1{2^{c_{p+1}}3^{d_{p+1}}5^{p+1}}-a_{p+1}\right|\\
				&\quad\geq
                \f1{30^{N}}-\f1{12\times30^{5N}}>\f1{12\times30^{5N}}
                                \, .	
			\end{align*}
			Hence $(\hat c_{p+1},\hat d_{p+1})=(c_{p+1}, d_{p+1})$, and
			\[
				\mu_{p+1}=\left|1-\f1{2^{c_{p+1}}3^{d_{p+1}}5^{p+1}}-a_{p+1}\right|\leq
                                \f1{12\times 30^{5N}}  \, . \qedhere
			\]
		\end{itemize}

                \end{proof}

		In other words, $\Minsf$ has to stay $\f1{12\times
                          30^{5N}}$-close to a faithful simulation of  $\Mcal$'s execution to avoid triggering a \CCE~punishment.       

		Now,  if $\Minsf$ ends the game at
                          step $p\leq N$, the final cost is then
			\begin{align*}
				64+1-a_p=64+\delta_p+\f1{2^{c_p}3^{d_p}5^p}&\geq
                                64+\delta_p+\f1{5^{2N}}\\
                                &\geq
                                64+\f1{5^{2N}}-\f1{12\times
                                  30^{5N}}\\&\geq 64+\f{11}{12\times
                                  30^{5N}} \, .
			\end{align*}
	To play more than $N$ steps, $\Minsf$ has to cheat  on a zero-test simulation. This will trigger a Case~$2$ punishment.
Consider, for instance\footnote{All other cases work similarly.},
		 the case in which $\Maxsf$ gets to move to a $\CNZ_1$ module while $c_p=0$. Then using \cref{prop:cnz}, $\Maxsf$ can secure a cost of 
					\begin{align*}
				&64+\min\enstqcolon{\left|\f1{2^{c}3^{d}5^{n}}-1+a_p\right|}{c,d,n\in\Nbb,0<c}\\
					&=64+\min\enstqcolon{\left|\f1{2^{c}3^{d}5^{n}}-\f1{3^{d_p}5^p}-\delta_p\right|}{c,d,n\in\Nbb,0<c}\\
					&\geq 64+\min\enstqcolon{\left|\f1{2^{c}3^{d}5^{n}}-\f1{3^{d_p}5^p}\right|}{c,d,n\in\Nbb,0<c}-\mu_p
                                   \, .
			\end{align*}

			Since $d_p\leq p$, we have
                        $2^{5p}>3^{3p}>5^{2p}>3^{d_p}5^p$, and thus
                        			\begin{align*}      			
				\min&\enstqcolon{\left|\f1{2^{c}3^{d}5^{n}}-\f1{3^{d_p}5^p}\right|}{c,d,n\in\Nbb,0<c}\\
				&= \min\enstqcolon{\left|\f1{2^{c}3^{d}5^{n}}-\f1{3^{d_p}5^p}\right|}{\begin{array}c
				c \in {1,\dots, 5p}\\
						d\in\{0,\dots,3p\}\\ n\in\{0,\dots,2p\}
				\end{array}} \, .
			\end{align*}
    Every element in this last set is a multiple of 
    $\f1{30^{5p}}$
                      and cannot be $0$ since $c>0$. Since
                      $p\leq N$, we can conclude that when $\Minsf$ cheats on a zero-test transition, 
			$\Maxsf$ can secure a cost of at least
			\[
				64+\f1{30^{5p}}-\f1{12\times
                                  30^{5N}}\geq64+\f{11}{12\times
                                  30^{5N}} \, . \qedhere
			\]

	\end{itemize}
\end{proof}


\section{Interlude: Region construction and some useful simplifications}
\label{sec-interlude}
\subsection{Restraining Clock Values to [0,1)}
\label{sec-zero-one}

Before presenting the well-known notions of regions and "region \wtgs", let us first restrict the setting to 
"$\intfo01$-\wtgs", which will simplify the "region" notations.

\begin{definition}
A ""$\intfo01$-\wtg"" is a "weighted timed game" where for every reachable configuration $(\ell,\nu)$, $0\leq \nu(x)<1$ for any clock $x$.
Similarly, a ""$\intff01$-\wtg"" is a "\wtg " where for every reachable configuration $(\ell,\nu)$, $0\leq \nu(x)\leq 1$ for any clock $x$.
\end{definition}

\begin{lemma}
\label{lem-zero-one}
For any "\wtg" $\Gcal$ with bounded clocks,
there is an equivalent "$\intfo01$-\wtg" $\Gcal'$.  
\end{lemma}
\begin{proof}
See \cite{bouyer2006}, Proposition 2 for a detailed proof. Their proof is for $1$-clock \wtg. However, the construction can easily be generalized to any number of clocks.

The intuition of the construction is that the information of the integer parts of the clocks can be contained in the locations, while the clocks keep track only of the fractional part:
Let $\mathcal{G}=(L_{\Minsf},L_{\Maxsf},G,\mathcal{X},T,w)$ and let us build $\mathcal{G'}=(L'_{\Minsf},L'_{\Maxsf},G,\mathcal{X},T',w')$.
First, w.l.o.g., let us assume that all clocks are bounded by an integer $M$ (\cite{behrmann2001}). Then, for $\Player\in\set{\Minsf,\Maxsf}$, let $L_\Player'=L_\Player \times M^{|\Xcal|}$, and define $T'$ and $w'$ such that a valuation $\nu=(x_1,\dots,x_{|\Xcal|})\in \intfo01^\Xcal$ in a location $(\ell,n_1,\dots, n_{|\Xcal|})$ in $\Gcal'$ is equivalent to a valuation $(n_1+x_1,\dots,n_{|\Xcal|}+x_{|\Xcal|})$ in location $\ell$ in $\Gcal$. 
Note that every transition of $\Gcal'$ with a guard $x=1$ for some clock $x$, resets $x$.
\end{proof}

\subsection{Regions and Region Trimmed Games}
\label{sec-region-trimmed}

\begin{definition}
	Let $\Xcal$ be a set of clocks.
	A \textbf{region} over $\Xcal$ is a tuple $r=(X_0,\dots,X_p,X_{=1})$ such that $X_i\neq \emptyset$ for all $1\leq i \leq p$, and
$\{X_0,\dots,X_p,X_{=1}\}$ is a partition of $\Xcal$: $\Xcal=\biguplus_{i=0}^pX_i \uplus X_{=1}$.

	We denote with $\intro*\Reg{\Xcal}$ the set of regions over $\Xcal$.
	A valuation $\nu\in\intff01^\Xcal$ is said to \textbf{belong} to region $r$, denoted with $\nu\sqsubset r$, whenever
	\begin{itemize}
		\item $\forall x\in\Xcal$, $\nu(x)=0\Leftrightarrow x\in X_0$,
		\item $\forall x\in\Xcal$, $\nu(x)=1\Leftrightarrow x\in X_{=1}$,
		\item $\forall x,y\in\Xcal$, $\ \nu(x)< \nu(y)<1\Leftrightarrow \exists i,j\in\{0,\dots, p\}\text{ s.t. } i<j\wedge x\in X_i\wedge y\in X_j$.
	\end{itemize}

For $r=(X_0,\dots, X_p, X_{=1})$ a "region" and $X$ a non-empty subset of $\Xcal$, we denote with $r[X:=0]$ the "region" $(X_0\cup X,X_1\setminus X,\dots, X_p\setminus X, X_{=1}\setminus X)$. In other words, $r[X:=0]$ is the "region" such that if $\nu$ belongs to $r$ then $\nu[X:=0]$ belongs to $r[X:=0]$.  We say that $r$ is \textbf{""$\intfo01$-region""} when $X_{=1}=\emptyset$, in which case we abuse notation and write $r$ as  $(X_0,\dots, X_p)$. We denote with $\intro*\Reglt{\Xcal}$ the set of "$\intfo01$-regions" over $\Xcal$.

A \textbf{""time-successor""} of a "region" $r$, with $r=(X_0,\dots, X_p)\in \Reglt{\Xcal}$  is a "region" $r'=(X'_0,\dots, X'_{p'}, X'_{=1})$ such that either $r'=r$; or $X_0'=\emptyset$ and 
$X'_i=X_{i-1}$ for $1\leq i \leq p$, and either $X_p=X'_{=1}$ or $X_p=X'_{p+1}$ (and  then $X'_{=1}=\emptyset$).

\end{definition}

We often abuse notation and write $r$ for the set of valuations $\nu\sqsubset r$ it represents.

\begin{definition}[Region \wtg \cite{BBL2008, Busatto2018}]
Let $\mathcal G = (L_\mathsf{Min}, L_\mathsf{Max},G,\mathcal{X},T,w)$ be a "$\intfo01$-\wtg".
The \textbf{""region construction""} $\intro*\RC{\mathcal G}$ of $\Gcal$ is the "$\intfo01$-\wtg" $\RC{\Gcal}=(L'_\mathsf{Min}, L'_\mathsf{Max},G',\mathcal{X},T',w')$ with
\begin{itemize}
	\item $L'_\Player=L_\Player\times\Reglt{\Xcal}$
	for $\mathsf P\in\{\Minsf,\Maxsf\}$.
	\item $G'=G\times\Reglt{\Xcal}$.
	\item For every $r=(X_0,\dots X_p)\in \Reglt{\Xcal}$, for every $r'=(X'_0,\dots, X'_{p'}, X'_{=1})\in \Reg{\Xcal}$ a "time-successor" of $r$,
	if $\ell \xrightarrow{C,X} \ell' \in T$ then
	$(\ell,r) \xrightarrow{C\cup C(r'),X} (\ell',r'[X:=0])$,
with $$C(r')=\setof{(x=0)}{x\in X'_0}\cup \setof{(x=1)}{x\in X'_{=1}}\cup \setof{(0<x<1)}{x\in X'_i, 1\leq i\leq p}.$$
	\item For $\ell\in L_{\Minsf}\cup L_{\Maxsf}$ and $r\in\Reglt{\Xcal}$, $w'(\ell,r)=w(\ell)$.
	\item For $t=(\ell,r) \xrightarrow{C\cup C(r'),X} (\ell',r'[X:=0])\in T'$, 
	$w'(t)=w\left(\ell \xrightarrow{C,X} \ell'\right)$.
\end{itemize}
\end{definition}

\begin{remark}[State complexity]
Restraining the clock values to $\intfo01$ and then applying the "region construction" yields a \wtg with the same number of locations
 exponential in the size of the input, like the one obtained by applying the "region construction" directly. 
\end{remark}

\AP While applying simplifying transformations to $\RC{\Gcal}$, we wish to preserve the \openquote one-region-per-location'' property. Thus let us formally define what is \emph{a} region \wtg (as opposed to \emph{the} region \wtg).
A \wtg $\Gcal$ is \textbf{""a region weighted timed game""} if there is a ""region assignment"" $\intro*\reg\cdot: L\to\Reglt{\Xcal}$, such that
for any transition $t:\ell\xrightarrow{C,X} \ell'$, the valuations $\nu+\delta$ with $\nu\sqsubset \reg{\ell}$ and $\delta\geq 0$ that satisfy $C$ are contained in a unique "region" $r$, such that $r[X:=0]=\reg{\ell'}$. Furthermore, for any initial configuration $(i,\nu)$, we require $\nu\sqsubset \reg{i}$.

	
	Obviously $\RC{\Gcal}$ is a "region \wtg " for any "$\intfo01$-\wtg" $\Gcal$.
	For any location $\ell$ in a "region" \wtg, let $\intro*\upclock{\ell}\eqdef X_p$ with $\reg{\ell}=(X_0,\dots X_p)$.

	Let us now ``trim'' $\RC{\Gcal}$, \ie delete every useless transition, and every useless guard on transitions:

	\begin{definition}[Trimmed region \wtg]
	\label{def-trimmed}
	A  region $\intfo01$-\wtg $\Gcal$ is \textbf{""trimmed""} if for any transition $t:\ell\xrightarrow{C,X} \ell'$  and any "region" $r=\reg{\ell}$ in $\Gcal$,
	\begin{itemize}
		\item for any valuation $\nu\sqsubset r$, there exists some $\delta\geq0$ such that $\nu+\delta\models C$.
		\item for any $c\in C$, there exist a valuation $\nu\sqsubset r$ and some $\delta\geq 0$ such that $\nu+\delta$ is in $\intff01^\Xcal$ and
		$\nu+\delta\not\models c$.
	\end{itemize}
\end{definition}

In other words, there are no inaccessible transitions from any location-region pair. Furthermore, there are no unnecessary guards in $C$ (the ones that are always verified from the region). Removing inaccessible transitions and unnecessary guards can always be done from any "region \wtg" without change in value.

Since every "$\intfo01$-\wtg" $\Gcal$ is equivalent to the "region" "$\intfo01$-\wtg" $\RC{\Gcal}$, and every "region \wtg" is equivalent to a "trimmed" "region \wtg", we can always assume $\Gcal$ to be a "trimmed" "region \wtg".

\begin{observation}
\label{lem-trimmed}
	For $\Gcal$ a "trimmed" "region" "$\intfo01$-\wtg", for any transition $t:\ell \xrightarrow{C,X} \ell'$, with $\reg{\ell}=(X_0,\dots, X_p)$:
	\begin{itemize}
		\item if $(y=0)\in C$ or $(y>0)\in C$ for some clock $y$ then $y\in X_0$. In other words, $y$ is $0$ on the whole region $r$.
		\item if $(y=1)\in C$ for some clock $y$, then $y\in X_p$. In other words, $y$ is one of the clocks with largest value in $r$.
		\item there cannot be both $x=0$ and $y=1$ in $C$ for any two clocks $x,y$.
	\end{itemize}
	Furthermore, $\Gcal$ has a global constraint $x\leq 1$ on every clock $x\in \Xcal$.
\end{observation}
\begin{proof}
	For the first point, notice that if $y$ was not $0$, then either the transition or the guard would have been trimmed. For the second one, if $y\notin X_p$ then upon taking the transition with "guard" $y=1$ there is a clock $x\in X_p$ such that $x>y=1$.
	For the third point, for any valuation in $r$, there is no $d$ such that $\nu + d$ satisfies both guards.
\end{proof}

\begin{definition}[Sub-\wtgs]
A ""sub-\wtg"" is a tuple $(L_\mathsf{Min}, L_\mathsf{Max},G,\mathcal{X},T,w, \wout)$ where $(L_\mathsf{Min}, L_\mathsf{Max},G,\mathcal{X},T,w)$ is a region \wtg, and $\wout$ is an ""output function"" which assigns  to each location in $G$ a piecewise linear function defined on the region of the associated location.
The weight of a play in a "sub-\wtg" is the sum of the weights of the delayed transition and of the cost given by the "output function" $\wout$.
\end{definition}
\AP
This notion allows us to study independently a part of a \wtg if we know the value function outside (or an approximation); this value function is given to the "sub-\wtg" as its "output functions".
We call an ""output edge"" any edge in a "sub-\wtg" leading to a goal location.


\section{Almost divergent weighted timed games}
\label{sec-dec}

One interesting property of the two-clock construction $\GMcal$ defined in 
\cref{sec-undec} is that $\GMcal$ is in general \emph{not} "almost non-Zeno".
\begin{definition}[Almost non-Zeno \wtg]
A \wtg $\Gcal$ with non-negative weights is \textbf{""almost non-Zeno""} if for any finite run $\rho$ in $\Gcal$ that follows a region cycle of $\RC{\Gcal}$, $\weight(\rho)=0$ or $\geq 1$.
\end{definition}
Indeed, cycle formed by transition modules can yield arbitrarily small weight, because of the location weight $30$ on control locations. 
This contrasts with the $3$-clock reduction in \cite{Bouyer2015}, which is almost non-Zeno. 

When generalized to \wtgs with arbitrary weights, almost non-Zenoness becomes almost divergence.
 \begin{definition}[Almost divergent \wtg (\cite{Busatto2018})]
  A strongly connected component (SCC) $S$ of $\mathcal R(\mathcal G)$ is said to be non-negative (resp. non-positive) if every run $\rho$ following a 
 cycle in $S$ satisfies either $\weight(\rho)\ge 1$ or $\weight(\rho) =0$ (resp. either $\weight(\rho)\le -1$ or $\weight(\rho) =0$).
    A WTG $\mathcal G$ is ""\textbf{almost divergent}"" if every SCC of $\mathcal R(\mathcal G)$ is either non-positive or non-negative.
    
 \end{definition}

 \begin{remark} 
Deciding if a \wtg is almost divergent is \textsc{PSpace}-complete (\cite{Busatto2018}, Theorem 4).
\end{remark}

In this section, we will show why it was impossible for $\GMcal$ to be "almost divergent":

\begin{restatable}{theorem}{thmDec}\label{thm-dec}
Given a turn-based two-clock almost divergent weighted timed game, the
Value Problem is decidable.
\end{restatable}
The proof of \cref{thm-dec} relies on the "partial unfolding" used in the approximability proof of \cite{Bouyer2015}, and on several techniques to turn a \wtg into an equivalent, simpler game with desirable properties, such as relaxing guards and adding clock resets to every transition.

\subsection{A few simplifications allowed by almost divergence}
\label{sec-wlog}

\begin{remark}[Removing infinite-value locations]
	
	Since reachability is decidable in timed games, it is decidable whether $\Value{\mathcal{G}}(\ell_0,\nu_0)=+\infty$ in polynomial time in the size of the explicit region graph. (This does not take into account infinite value created by arbitrarily large delays, we will see in the next remark that we can assume that clock valuations, and thus delays, are bounded).
	
	However, deciding whether a configuration $(\ell, \bold{0})$ has valuation $-\infty$ is undecidable in the general case (\cite{BusattoGastonMR23}, Proposition 3.12). Fortunately, it is decidable and \textsc{ExpTime}-complete for "almost divergent" \wtgs (\cite{BusattoGastonMR23}, Proposition 3.11). 
	Furthermore, if a configuration $(\ell,\nu)$ has infinite value in $\Rcal(\Gcal)$, then any configuration in $\ell$ has infinite value.

Therefore, either the initial location is a $\pm\infty$ location, or any location of infinite value can be removed without change in value in the remaining locations. Indeed, from finite-value location, $\Minsf$ (resp. $\Maxsf$) can avoid any location of value $+\infty$ (resp. $-\infty$).
\end{remark}

\begin{remark}[Bounded clocks]
Any "almost divergent" \wtg $\Gcal$ is equivalent to an "almost divergent" \wtg with bounded clocks. The idea is to reset clock automatically once they are above $M$ the maximum of all integers appearing in guards. The main problem of this construction is the following: if $\Maxsf$ has access to an unguarded transition from a location $\ell$, which he could take with arbitrarily large delay, this would be transformed into a self-loop. 
With "almost divergence", we can fix this problem at the root by guarding any transition in $\Gcal$ before applying the transformation:
If $w(\ell)\leq 0$, then $\Maxsf$ has no incentive to spend an arbitrarily long delay in $\ell$ in $\Gcal$, thus adding a guard bounding the delay taken in $\ell$ by $M$\footnote{Delays can be bounded by adding a clock. Here the goal is to bound the delay, then apply the bounded delay transformation which will not create pathological self-loops. The additional clock can then be removed.}. If $w(\ell) > 0$, then the Value in $\ell$ is either $+\infty$, or $-\infty$, depending on whether $\Minsf$ can take a negative-weight cycle arbitrarily often afterward or not. However, thanks to the previous remark, locations of $-\infty$ value can be detected and removed without loss of generality. If $\ell$ is a $+\infty$ location, then transforming it into a self-loop controlled by $\Maxsf$ does not change the Value.

\end{remark}

\begin{remark}[No singleton region in cycles]
If a run visits twice a singleton-region location, then it visits twice the exact same configuration. The run in between the two visits yields either a negative or a non-negative weight, thus we can assume that one player has an incentive not to visit this location twice. Since we removed infinite-value location, this player is able to evade this cycle. Therefore we can unfold a \wtg such that no singleton-region location is visited twice, by ensuring that an attempt at second visit leads to an infinite-value location punishing the player that should have evaded it.
This approach is similar to the one in \cite{monmege2024journal}, whose construction aimed to eliminate reset transition in one-clock \wtg. The main difference is on how we identify which player should be punished: in an "almost divergent" \wtg, SCC are either non-negative or non-positive, and simple cycles are either negative, null or positive. Thus, after visiting a singleton-region location, one only need to take note of if they visited either a negative or a positive simple cycle at some point.
\end{remark}

\textbf{In the remainder of this paper, every weighted timed game is turn-based, almost divergent, of value in $\mathbb{R}$, with bounded clocks and delays.}.

\subsection{Unfolding Almost Divergent Weighted Timed Games}

%
%

Bouyer \etal \cite{Bouyer2015} showed that, even though the "Value Problem" is undecidable for \wtg with non-negative weight and $3$ clocks or more, it is approximable in the subclass of "almost non-Zeno \wtg".

Busatto-Gaston \etal \cite{Busatto2018} extended this approximability result to "almost divergent \wtg" with arbitrary weights.



In this section, we use the structure of the approximability proofs of \cite{Bouyer2015,Busatto2018} to prove decidability for "almost divergent \wtgs" with $2$ clocks and arbitrary weights.
Let us first assume that $\mathcal G$ is an SCC, either non-positive or non-negative. If we are able to compute the value of an SCC, then it will be easy to string them together into more complex \wtgs.

In an acyclic "\wtg", the "value" is decidable and can be computed  from the target locations up to the initial location, by computing for each node $\ell$ a function $W_\ell:\reg{\ell}\to\Rbb$ which assigns to a valuation $\nu\in \reg{\ell}$ the optimal weight 
$\Value{\Gcal}(\ell,\nu)$.
By construction, every $W_\ell$ is a piecewise linear function.

   \AP A cycle is called a ""zero-cycle"" if every run $\rho$ following has weight $0$.
Intuitively, we will unfold cycles of weight $\geq 1$ (resp. $\leq -1$) to obtain a ``tree-like'' \wtg where only the only cycles are "zero-cycles"; we will deal with them separately.

\paragraph*{Semi-unfolding}
 For any "\wtg" $\Gcal$, let $\intro*\Gtilde$ be the ""semi-unfolded \wtg"" built from $\RC{\Gcal}$
  in \cite{Busatto2018}
  \footnote{Note that the semi-unfolding of \cite{Busatto2018} is quite similar to the semi-unfolding of \cite{Bouyer2015}.}.
 
   \AP First colour in ""green"" every location and edge that are part of a "zero-cycle". 
We define the \textbf{""kernel""} $\Kcal$ of $\Gcal$ as the restriction of $\RC{\Gcal}$ to "fully-green" strongly
connected components. Edges that leave $\Kcal$ are called the output edges of $\Kcal$.
According to Busatto-Gaston \etal~\cite{Busatto2018}, "zero-cycles" are stable by rotation (if $\pi\pi'$ is a "zero-cycle", then so is $\pi'\pi$), thus any cycle of $\Kcal$ is a "zero-cycle" .

Then we partially unfold $\RC{\Gcal}$ into a finite tree structure $\intro*\TC{\Gcal}$ :
starting from the initial location $i$ as a root, we follow every possible path in $\Gcal$, with a node for each time we visit a (non-"kernel") location, to avoid creating cycles.
However,
when along a branch we enter the "kernel" in some location $\ell$, we create a node $\Kcal_\ell$ instead of $\ell$, and for each "output edge" $t$ of $\Kcal$ accessible from $\ell$, with $t$ leading to a location $\ell'$, let $\ell'$ (or $\Kcal_\ell'$ if $\ell'\in\Kcal$) be a child of $\Kcal_\ell$, and continue to unfold from there.

We stop unfolding when, along any branch, a location or transition outside of $\Kcal$ is visited at
least $\gamma$ times, where $\gamma$ is the bound given in \cite{Busatto2018}\footnote{See Section 5 of \cite{Busatto2018}. $\gamma$ is polynomial in $w$ and $|\Rcal(\Gcal)|$.}.

To obtain $\Gtilde$ from $\TC{\Gcal}$, replace each node $\Kcal_\ell$ by a copy of the strongly connected component of $\Kcal$ that contains $\ell$
(see \cite{Busatto2018} for the formal construction). Then $\Value{\Gtilde}(i,\nu)=\Value{\Gcal}(i,\nu)$ for any $\nu\in\reg{i}$.


In the partially unfolded games $\Gtilde$ with three clocks or more, the cause of undecidability is inside the "kernel" nodes.
Hence, in \cite{Bouyer2015,Busatto2018}, the approximation happens in the "kernels". However, with only two clocks, the value is decidable in "kernels":

\begin{definition}[Kernel weighted timed games]
A \textbf{""kernel weighted timed game""} $\Gcal$ is a trimmed region "$\intfo01$-""sub-\wtg "
$(L_{\Minsf}, L_{\Maxsf},G,\mathcal{X},T,w,\wout)$ whose non-goal locations form a SCC, and such that every run following a region cycle is of weight $0$, and where, for any target location $\ell\in G$, the "output weight function" $\wout(\ell,\cdot):\reg{\ell}\to \Reel$ is continuous and piecewise linear.
\end{definition}

\begin{restatable}{theorem}{thmKernel}
\label{lem-linear-kernel}
For any two-clock "kernel \wtg" $\Gcal$, for  any location $i$, $W_i$ is a continuous piecewise linear function which can be computed through the "value iteration algorithm".
\end{restatable}

This is the main technical result at the root of our decidability result; it will be proven in \cref{sec-value-kernel}. Let us show first how \cref{lem-linear-kernel} entails Value decidability of the partial unfolding $\Gtilde$:

\begin{lemma}
For every node $n$ in the tree $\TC{\Gcal}$, one can compute a continuous piecewise linear function $W_n$ such that $W_n: \nu\mapsto \Value{\Gtilde}(\ell,\nu)$, where $\ell$ is either $n$ if $n\not\in\Kcal$, or the entrance location of $n=\Kcal_\ell$.
\end{lemma}
\begin{proof}
In the tree structure of $\Gtilde$, consider a node $n$:
if $n$ is a leaf, then if $n \in G$, let $W_n$ be the constant null function, otherwise let $W_n$ be the constant $+\infty$ function.
Now consider the case in which $n$ is not a leaf, and by induction hypothesis assume that for every child $n'$ of $n$, $W_n'$ is continuous and piecewise linear.
If $n=\ell \not\in\Kcal$ then 
\[W_n:\nu \mapsto \underset{\nu+d\models C}{\underset{t:n \xrightarrow{C,X} n' \in \TC{\Gcal}}{\inf}}
W_{n'}((\nu+d)[\reset{X}])
+d w(\ell)+w(t)
\quad \text{(if $\ell\in\Maxsf$, replace $\inf$ by $\sup$).}\] 

Thus by induction $W_n$ is also continuous and piecewise linear. 

Otherwise, $n=\Kcal_\ell$ for some $\ell\in \Kcal$. Let $K$ be the SCC containing $\ell$, and $T_{\mathsf{out}}$ the "output edges" leaving from $K$. Consider the "kernel" \wtg
$K_\ell=(L'_{\Minsf},L'_{\Maxsf} ,G',\mathcal{X},T',w',\wout)$ where
\begin{itemize}
\item $G'=\setof{\ell_t}{t \in T_{out}}$.
\item For every $\Player$, $L'_\Player = K \cap L_\Player$.
\item $T'=T_{|K\times K}\cup \set{\ell'\xrightarrow{C,X} \ell_t \,|\, t:\ell'\xrightarrow{C,X} \ell'' \in T_{out}}$
\item for every $t:\ell'\xrightarrow{C,X} \ell'' \in T_{\mathsf{out}}$, $\wout(\ell_t,\cdot):\nu\mapsto W_{\ell''}(\nu)$, which is piecewise linear by induction hypothesis.
\item $w'=w_{|K\cup T'}$.
\end{itemize} 
Then $W_n=\Value{K_\ell}(\ell,\cdot)$, which is piecewise linear and can be computed according to \cref{lem-linear-kernel}.
\end{proof}

This is sufficient to conclude the proof of \cref{thm-dec}.

\thmDec*

\subsection{Simplifying Transformations of Kernel Games}
\label{sec-all-simpl}
Before proving \cref{lem-linear-kernel}, let us apply some useful simplifying transformations that preserve the value. These transformations happen in three steps:
\begin{description}
\item[Step $1$:] Erase weights (by sweeping them into the output functions).
\item[Step $2$:] Relax every strict guard into a strict-or-equal guard.
\item[Step $3$:] Transform the resulting weightless "relaxed" "trimmed"  "kernel \wtg" such that every transition resets at least one clock.
\end{description}

\paragraph*{Simplification dependencies:}

Step $2$ relies on the trimming operation defined in \cref{sec-region-trimmed}: without it, Step $2$ would create pathological cases where relaxing some "guards" would allow a player to take transitions that would have been unreachable in the original \wtg.

Relaxing guards in step $2$ is a technique that has merits of its own outside of the scope of the proof. For instance, the value of a \wtg might not be reached by an optimal strategy for $\Minsf$, but could be the infimum produced by a set of strategies $\epsilon$-close to an optimal strategy in the relaxed \wtg.

Relaxing guards is also a prerequisite for step $3$, which is the most important of all four steps. In a two-clock "\wtg", resetting at least one clock in each transition allows us to consider only regions in one dimension. Reducing a two-dimensional problem to a one-dimensional one is key to the termination argument in \cref{sec-value-kernel}.

\subsubsection{Erasing weights from a kernel}

Let $\Gcal=(L_{\Minsf},L_{\Maxsf} ,G,\mathcal{X},T,w,\wout)$ be a "kernel \wtg". Without loss of generality, we assume that for each goal location $\ell_G$, there is exactly one transition leading to $\ell_G$, and that this transition does not reset any clock.

First let us focus on transitions that do not lead to a goal location.

\begin{lemma}[Uniqueness of the weight in a "zero-cycle"]\label{lem:wunique0}
Take two non-goal locations $\ell_1,\ell_2$ of $\Gcal$, and $\nu\in\reg{\ell_2}$.
There is a weight $W \in \mathbb R$ such that every run starting in the location $\ell_1$ (with any valuation)
and ending in the configuration $(\ell_2,\nu)$ has a weight $W$.
\end{lemma}

\begin{proof}
Let us take a path $\pi_1$ from $\ell_1$ to $\ell_2$ 
and a path $\pi_2$ from $\ell_2$ to $\ell_1$ ($\ell_1$ and $\ell_2$ are part of an SCC, so these paths exist).
Then $\pi_1\pi_2$ is a "zero-cycle".

Let us take a run $\rho_2$ following $\pi_2$ and starting at the configuration $(\ell_2,\nu)$.
Then for all run $\rho_1$ following $\pi_1$ and ending in the configuration $(\ell_2,\nu)$,
$\rho_1\rho_2$ follows $\pi_1\pi_2$, which gives 
$\weight(\rho_1)=-\weight(\rho_2)(=W)$ (by definition of a "zero-cycle").
\end{proof}

This entails that, considering a transition from a location with non-zero weight, fixing the first or second valuation immediately fixes the delay:

\begin{lemma}[Uniqueness of the delay]\label{lem:delunique0}
Consider a transition $t=\ell_1\xrightarrow{C,X} \ell_2$ in $\Gcal$, such that $w(\ell_1)\neq 0$ and $\ell_2 \notin G$.
For all $\nu_2 \in \reg{\ell_2}$ (resp. $\nu_1 \in \reg{\ell_1}$), there exists a unique delay $d \in [0,1]$ such that there exists a valuation $\nu_1 \in \reg{\ell_1}$  (resp. $\nu_2 \in \reg{\ell_2}$)
such that $\nu_1+d \models C$ and $(\nu_1+d)[X:=0]=\nu_2$.
\end{lemma}
\begin{proof}
As a "trimmed" "region \wtg" does not contain transitions where the guard cannot be satisfied, such delay exists.
Let us take a valuation $\nu_2 \in \reg{\ell_2}$.
By lemma~\ref{lem:wunique0},
there exists a $W\in\Rbb$ such that, for every $\nu_1\in\reg{\ell_1}$, every path from $(\ell_1,\nu_1)$ to $(\ell_2 ,\nu_2)$ has weight $W$.
If there is a valuation $\nu_1$ such that $(\ell_1,\nu_1)\xrightarrow{d,t}(\ell_2,\nu_2)$, we obtain $d\cdot w(\ell_1)+w(t)=W$.
As $w(\ell_1)\neq 0$, there is only one possibility for the delay $d$.

Similarly, consider a path $\rho:(\ell_2,\nu) \rightarrow^* (\ell_1,\nu_1)$ of weight $W$. Then a delayed transition $(\ell_1,\nu_1) \xrightarrow{d,t} (\ell_2,\nu_2)$ must be of weight $-W$, which fixes $d$. 
\end{proof}

\cref{lem:delunique0} puts a heavy constraint on the regions of $\ell_1$ and $\ell_2$. Let us define the following quantity, which is conserved by unique-delay transitions:

\begin{definition}[Circular clock difference]
Let $\reg{\ell}$ be a horizontal or vertical open segment.
Then the ""\textbf{circular clock difference}"" $\intro*\DDelta$ of a valuation $\nu\sqsubset\regadh{\ell}$ is defined as: 
\[\DDelta(\nu)=\begin{cases}
y \text{ when } \nu=(0,y)\text{ with $y\geq 0$,}\\
1-x \text{ when } \nu=(x,0) \text{ with $x>0$.}
\end{cases}\]
Given $q\in (0,1)$, we denote $R_q$  the set of valuations of clock difference $q$.
\end{definition}

\begin{corollary}[Region of a location (with non-zero weight) in a "zero-cycle"]\label{lem:region0}
\label{lem:guard0}
Let $t=\ell_1\xrightarrow{C,X} \ell_2$ be a transition in $\Gcal$, such that $w(\ell_1)\neq 0$ and $\ell_2 \notin G$. Then
\begin{enumerate}
\item $\reg{\ell_1}$ and  $\reg{\ell_2}$ are horizontal or vertical open segments or singleton
(i.e. $\exists z \in \mathcal X, \forall \nu \sqsubset r,\nu(z)=0$),
\item $C$ contains a guard $\{z_1=0\}$ or $\{z_1=1\}$
(i.e. from a fixed valuation, the set of possible delays is a singleton).
\item we can compute $a,b\in\mathbb{Z}$ such that any delayed transition
$(\ell_1,\nu_1)\xrightarrow{d,t}(\ell_2,\nu_2)$ has a weight $a\DDelta(\nu_2)+b$. 
We denote this linear function by $\tilde{w}_t$.
\end{enumerate}
\end{corollary}

\begin{proof}
The first two points are direct consequences of \cref{lem:delunique0}.
Item (2) entails that for any delayed transition $(\ell_1,\nu_1)\xrightarrow{d,t}(\ell_2,\nu_2)$, $d$ is either $0$, $\DDelta(\nu_2)$ or $1-\DDelta(\nu_2)$, which suffices to prove Item (3).
\end{proof}

A similar observation can be made for output transitions.
\begin{lemma}[Weight of an output transition inducing continuous weight]
\label{lem-linear-out}
Let $t=\ell_1\xrightarrow{C,X}\ell_2$ such that $w(\ell_1)\neq 0$, with 
$\ell_2\in G$. Then we can compute $a,b\in\mathbb{Z}$ such that any delayed transition
$(\ell_1,\nu_1)\xrightarrow{d,t}(\ell_2,\nu_2)$ has a weight $a\nu_2(z_2)+b$ where $z_2$ is the smallest clock in $\reg{\ell_1}$.
\end{lemma}
\begin{proof}
According to \cref{lem:region0}, the region of $\ell_1$ is an open segment.
Remember that any output edge does not reset clocks. Thus, for any delayed transition $(\ell_1,\nu_1)\xrightarrow{d,t}(\ell_2,\nu_2)$, $d$ is equal to $\nu_2(z_2)$.
\end{proof}

\AP For any transition $t$ leaving a location $\ell_1$, we say that $t$ ""induces continuous  weight"" if $w(\ell_1)\neq 0$ and $t$ does not carry a guard of the form $z=0$.

\AP We say that a transition $t=\ell_1\xrightarrow{C,X}\ell_2$ is ""deleting information"" if, for any valuation $\nu_2\in\reg{\ell_2}$, the set of predecessors of $\nu_2$ through $t$ is not a singleton.
For instance, any transition that can change the "circular clock difference" is deleting information.

\begin{lemma}[Characterization of "zero-cycle"s]\label{lem:characterization}
A $0$-cycle is either:
\begin{itemize}
\item a cycle where no transition "induces continuous weight". The weight of a path following part of such a cycle depends only on discrete transition weights.
\item a cycle with no transition "deleting information". Following such a cycle does not change the "circular clock difference".
\end{itemize}

\end{lemma}
\begin{proof}
By contradiction, consider a cycle $\pi$ containing both a transition "deleting information" and a transition "inducing continuous weight".
Let $\pi=\ell_0 \xrightarrow{C_0,X_0} \cdots \ell_n \xrightarrow{C_n,X_n} \ell_0=t_0 \cdots t_n$.
We rotate $\pi$ such that $t_0$ is a transition "inducing continuous weight", and there is $0<j\le n$ such that
$t_j$ is a transition "deleting information", and all transitions $t_i$
 for $0<i<j$ are neither "inducing continuous weight" nor "deleting information".
Let us take a run $\rho=(\ell_0,\nu_0)\xrightarrow{d_0,t_0}\cdots \xrightarrow{d_j,t_j}(\ell_{j+1},\nu_{j+1})$ following $t_0\dots t_j$.

By \cref{lem:region0}, $\reg{\ell_0}$ is an open horizontal or vertical segment.
W.l.o.g. we assume that $\nu_0$ is of the form $(0,q)$ with $0<q<1$, and $C_0$ is $y=1$.
Since $t_2,\dots,t_{j-1}$ are not "deleting information", $\DDelta(\nu_j)=\DDelta(\nu_1)=q$.
Then $d_0=1-q$, and $\weight(\rho)=(1-q)\cdot w(\ell_0)+\sum_{0\le i\le j}w(t_i)$.

As $t_j$ is a transition "deleting information", there exists a valuation $\nu_0'\neq \nu_0$ 
and a path $\rho'$ from $(\ell_0,\nu_0')$ to $(\ell_j,\nu_j)$, such that 
$\weight(\rho')\neq \weight(\rho)$, in contradiction with lemma~\ref{lem:wunique0}.
\end{proof}
\begin{lemma}[Piecewise linearity]\label{lem:piecewise}

Let us take an acyclic path $\pi=\ell_0\xrightarrow{C_0,X_0} \ell_1 \xrightarrow{C_1,X_1} \cdots \xrightarrow{C_{n-1},X_{n-1}} \ell_{n}=t_0\dots t_{n-1}$ with $\ell_n\not\in G$. \\
Then we can compute a linear function $\tilde{w}_\pi$ such that for any $\rho$ following $\pi$ ending in $(\ell_n,\nu)$, $\weight(\rho)=\tilde{w}_\pi(\nu)$.
\end{lemma}

\begin{proof}
Let $I\subseteq [0,n-1]$ be the set of indices $i$ such that $t_i$ is 
a transition "inducing continuous weight".
If $I=\emptyset$, then 
$\tilde{w}_\pi(\nu)=\sum_{0\le i <n}w(t_i)$.
Otherwise, $\pi$ is part of a $0$-cycle, which contains no transition "deleting information", according to \cref{lem:characterization}.
Thus for any $\rho$ following $\pi$ ending in $(\ell_n,\nu)$,
$$\weight(\rho)=\sum_{i\in [0,n-1]\setminus I}
w(t_i)+\sum_{i\in I} \tilde{w}_{t_i}(\DDelta(\nu))\;.$$
\end{proof}

Now we can combine \cref{lem:piecewise,lem-linear-out}.
\begin{corollary}[Piecewise linearity]\label{cor:piecewise}
Let us take an acyclic path $\pi=\ell_0\xrightarrow{C_0,X_0} \ell_1 \cdots \ell_n\xrightarrow{C_{n},\emptyset} \ell_G=t_0\dots t_{n}$ with $\ell_G\in G$. \\
Then we can compute a linear function $\tilde{w}_\pi$ such that for any $\rho$ following $\pi$ ending in $(\ell_n,\nu)$, $\weight(\rho)=\tilde{w}_\pi(\nu)$.
\end{corollary}

Then, following \cref{lem:wunique0}, we can abstract this result from a specific path:
\begin{corollary}
\label{final-weight}
For any location $\ell_0$, and any goal location $\ell_G\in G$, there exists a linear function $\tilde{w}_{\ell_0,\ell_G}$ such that for any path $\rho$ starting in $\ell_0$ and ending in $\ell_G$ with valuation $\nu$, $\weight(\rho)=\tilde{w}_{\ell_0,\ell_G}(\nu)$.
\end{corollary}


We can now sweep these affine functions under the rug, or rather into the output functions of the kernel, to obtain a kernel without weights.
\begin{theorem}[Erasing kernel weights]
For any two-clock "kernel \wtg", we can build an equivalent two-clock "kernel \wtg" where all transitions and locations have weight $0$.
\end{theorem}

For each final location $\ell_G\in G$, let
$\tilde{w}_{\ell_0,\ell_G}$ be the function defined in \cref{final-weight}, which is piecewise linear and continuous.

\par We consider the kernel $\Gcal'=(L_\Minsf,L_\Maxsf,G,\mathcal X,T,w_0,w'_{out})$ where
\begin{itemize}
\item $w_0$ maps every location or transition to $0$.
\item $w'_{out}(\ell_G,\nu)=w_{out}(\ell_G,\nu)+\tilde{w}_{\ell_0,\ell_G}(\nu)$ for all $\ell_G\in G, \nu \in \reg{\ell_G}$.
\end{itemize}

Since $w'_{out}$ is piecewise linear and continuous over regions, $\mathcal G'$ is a kernel game.
Then, according to \cref{lem:piecewise}, $\Value{\Gcal'}(\ell_0,\cdot)=\Value{\Gcal}(\ell_0,\cdot)$.

\subsubsection{Relaxing Strict Guards}
\label{sec-relax}

A weightless "kernel \wtg" can easily be transformed into a weightless "kernel \wtg" without strict guards, without change in "value". 

\begin{restatable}{theorem}{thmspicy}
\label{thm-sim-spicy}
Let $\Gcal_\prec=(L_\Minsf^\prec, L_\Maxsf^\prec,G,\Xcal,T^\prec, \wout^\prec)$ be a weightless "trimmed" "region" "kernel \wtg". Let $\Gcal_\preceq = (L_\Minsf^\preceq, L_\Maxsf^\preceq,G,\Xcal,T^\preceq,\wout^\preceq)$ be a copy of $\Gcal_\prec$ where 
\begin{itemize}
\item every "guard" has been relaxed, i.e., every "guard" of the form $x>0$ and $x<1$ has been replaced by $x\geq 0$ or $x\leq 1$, respectively, for all $x\in\Xcal$.
\item  For any $\ell\in G$, the "output function" ${\wout}^\preceq(\ell,\cdot)$ is $\wout^\prec(\ell,\cdot)$ extended continuously to $\regadh{\ell}$ in $\Gcal_\preceq$.
\end{itemize}
Then $\Value{\Gcal_\prec}=\Value{\Gcal_\preceq}$.
\end{restatable}

Let us prove this with a bisimulation argument, similar to the approach used in \cite{brihaye2022}.

\begin{definition}[Simulation and bisimulation]

Let $\mathcal{G}=(L_\Minsf,L_\Maxsf,G,\mathcal{X},T,w)$ and $\mathcal{G}'=(L'_\Minsf,L'_\Maxsf,G',\mathcal{X},T',w')$ be \wtgs. Then a relation $R\subseteq (L_\Minsf\times \Rbb_{\geq 0}^\Xcal\times L'_\Minsf\times \Rbb_{\geq 0}^\Xcal) \cup (L_\Maxsf\times \Rbb_{\geq 0}^\Xcal\times L'_\Maxsf\times \Rbb_{\geq 0}^\Xcal)$ is a \textbf{simulation relation} when
\[(\ell,\nu)\,R\,(\ell',\nu') \implies \begin{cases}
\medskip
\text{$\ell\in G$ and $\ell'\in G'$,}\\
\text{or, for all $t\in T$, $d\in\Rbb$, there exist $t'$ and $d'\in\Rbb_{\geq 0}$ such that if }\\
 \text{$(\ell,\nu)\xrightarrow{d,t}(l_2,\nu_2)$ then $(\ell',\nu')\xrightarrow{d',t'}(\ell'_2,\nu'_2)$ where $(\ell_2,\nu_2)\,R\,(\ell'_2,\nu'_2)$.}
\end{cases}\]

$R$ is a \textbf{bisimulation} if $R$ and its converse are both simulations.
\end{definition}

\begin{example}
 The region relation $R$ between configurations of the same \wtg, where $(l_1,\nu_1)R(l_2,\nu_2)$ iff $l_1=l_2$ and $\nu_1,\nu_2$ are in the same region, is a bisimulation.
\end{example}

A (bi)simulation relation $R$ can be extended to runs:
$\rho_1 \,R\, \rho_2$ when $|\rho_1|=|\rho_2|$, and  $\rho_1^\Ccal(n) \,R\, \rho_2^\Ccal(n)$ for all $n\leq|\rho_1|$~\footnote{
For any run $\rho:c_0\rightarrow \dots \rightarrow c_n$, we denote $c_i$ with $\rho^\Ccal(i)$.}.

\begin{lemma}
\label{lem-bisim-WTG}
Let $\Gcal_1$ and $\Gcal_2$ be two \wtgs. 
Let $R$ be a bisimulation relation between $\Gcal_1$ and $\Gcal_2$, such that $\Gcal_1$ and $\Gcal_2$ start from configurations $c_1 \,R\, c_2$. Assume that $\Value{\Gcal_i} \in\mathbb{R}$ for $i\in \{1,2\}$. 
Then 

$$|\Value{\Gcal_1}(c_1)-\Value{\Gcal_2}(c_2)|\leq \sup
\enstq{|\weight(\rho_1)-\weight(\rho_2)|}{\begin{array}{c}
		\rho_1,\rho_2 \text{ finite plays of } \Gcal_1,\Gcal_2\\
		 \text{ starting in } c_1,c_2\\ 
		\rho_1 \,R\, \rho_2\\
		\end{array}}
\;.$$

\end{lemma}
\begin{proof}
For every strategy $\sigma_1^\Mimin$ for Min on $\Gcal_1$ and $\sigma_2^\Mamax$ for Max on $\Gcal_2$, there exist  $\sigma_2^\Mimin$ and $\sigma_1^\Mamax$ strategies for Min and Max on $\Gcal_2$ and $\Gcal_1$ respectively, such that $\rho_1 \,R\, \rho_2$ for $\rho_i=\play(c_i,\sigma_i^\Mimin,\sigma_i^\Mamax)$.
(In any Min location, $\sigma_2^\Mimin$ simulates $\sigma_1^\Mimin$ following $R$. Similarly in any Max location, $\sigma_1^\Mamax$ simulates $\sigma_2^\Mamax$ following $R$.)

Note that $\sigma_1^\Mamax, \sigma_2^\Mimin$ and $\rho_1,\rho_2$ are functions of $\sigma_1^\Mimin, \sigma_2^\Mamax$. However, we omit these arguments to lighten notations.

Since the set of strategies obtained by such a simulation is included in the set of all strategies,
\begin{align*}
\sup_{\sigma_2^\Mamax} \weight(\rho_1) \leq \sup_{\sigma^\Mamax} \Value{\Gcal_1}(\sigma_1^\Mimin, \sigma^\Mamax)\text{ for a fixed $\sigma_1^\Minsf$,}\\
\shortintertext{and}
\inf_{\sigma_1^\Mimin} \weight(\rho_2) \geq \inf_{\sigma^\Mimin} \Value{\Gcal_2}(\sigma^\Mimin, \sigma_2^\Mamax)\text{ for a fixed $\sigma_2^\Maxsf$.}
\end{align*}

Therefore \[\inf_{\sigma_1^\Mimin} 
\sup_{\sigma_2^\Mamax}
\weight(\rho_1)\leq 
\inf_{\sigma^\Mimin}\sup_{\sigma^\Mamax} \weight(\play_{\Gcal_1}(c_1,\sigma^\Mimin, \sigma^\Mamax))= \Value{\Gcal_1}\]
and
\begin{align*}
\inf_{\sigma_1^\Mimin}
\sup_{\sigma_2^\Mamax} 
\weight(\rho_2)
&\geq
\sup_{\sigma_2^\Mamax} 
\inf_{\sigma_1^\Mimin}
\weight(\rho_2)\\
&\geq 
\sup_{\sigma^\Mamax}\inf_{\sigma^\Mimin} \weight(\play_{\Gcal_2}(c_2,\sigma^\Mimin, \sigma^\Mamax))= \Value{\Gcal_2}
\end{align*}

Therefore $\Value{\Gcal_2}-\Value{\Gcal_1}\leq 
\sup\setof{\weight(\rho_2)-\weight(\rho_1)}{\rho_1 \,R\, \rho_2}\;.$

%

By a symmetric argument, one can obtain
\[\Value{\Gcal_1}-\Value{\Gcal_2}\leq 
\sup\setof{\weight(\rho_1)-\weight(\rho_2)}{\rho_1 \,R\, \rho_2}\;\]

%
and combine to conclude.

%
%
\end{proof}

Now let us define the following relation between valuations:
For $\epsilon>0$, two valuations $\nu,\nu'$ are \textbf{$\epsilon$-neighbours} if there exist $\epsilon_1,\epsilon_2\geq 0$ such that $\epsilon_1+\epsilon_2<\epsilon$, and for any clock $x\in\Xcal$, $\nu(x)-\nu'(x)\in[-\epsilon_1,\epsilon_2]$.

Let $\Gcal_\prec=(L_\Minsf^\prec, L_\Maxsf^\prec,G,\Xcal,T^\prec, \wout^\prec)$ be a weightless trimmed region weighted timed game. We consider $\Gcal_\preceq = (L_\Minsf^\preceq, L_\Maxsf^\preceq,G,\Xcal,T^\preceq,\wout^\preceq)$ a copy of $\Gcal_\prec$ where every strict guard has been relaxed into a strict-or-equal guard.

\begin{restatable}{definition}{defAdh}
Let $r=(X_0,\dots,X_p)$. Then the \textbf{""adherence""} of $r$, denoted by $\intro*\adh{r}$, is the set of regions
of the form $(Y_0,\dots,Y_{p'},Y_{=1})$ with $p'\leq p$ and $\iota: \intff0{p'} \to \intff0p$ strictly increasing such that and $Y_0= X_0 \cup \dots \cup X_{\iota(0)}$ and $Y_i= X_{\iota(i-1)+1}\cup \dots \cup X_{\iota(i)}$ for all $1\leq i \leq p'$ and $Y_{=1}= X_{\iota(p')+1}\cup \dots \cup X_{p}$ (or $\emptyset$ if $\iota(p')=p$) .
\end{restatable}
Let us abuse notation and write $\nu\sqsubset \adh{r}$ when $\nu\sqsubset {r'}$ for  $r'\in \adh{r}$.

\begin{lemma}
\label{lem-bisim-spicy-all}

For any configurations $(\ell_\prec,\nu_\prec)$ in $\Gcal_\prec$, 
and $(\ell_\preceq,\nu_\preceq)$ in $\Gcal_\preceq$, 

Let $(\ell_\prec,\nu_\prec) \,R_\epsilon\, (\ell_\preceq,\nu_\preceq)$ \tiff 
\begin{itemize}
\item $\ell_\prec=\ell_\preceq$ 
\item $\nu_\prec\sqsubset \reg{\ell_\prec}$ and $\nu_\preceq\sqsubset \regadh{\ell_\prec}$
\item $\nu_\prec$ and $\nu_\preceq$ are $\epsilon$-neighbours
\end{itemize}
Then $R_\epsilon$ is a bisimulation relation.
\end{lemma}

\begin{proof}
First observe that one side of the bisimulation is easier than the other one, since transitions in $\Gcal_\preceq$ are more permissive.
Therefore we will detail only the other side:

Let $t:\ell_1\xrightarrow{C,X} \ell_2$. Let $\overline{C}$ be the relaxation of $C$. Let $\nu_1$ be a valuation belonging to $\reg{\ell_1}$ in $\Gcal_\prec$ and let $\nu_1'$ be an $\epsilon$-neighbour of $\nu_1$. By definition, there are $\epsilon_1,\epsilon_2\geq 0$ such that $\epsilon_1+\epsilon_2<\epsilon$ and for any clock $x\in\Xcal$, $\nu_1(x)-\nu_1'(x)\in[-\epsilon_1,\epsilon_2]$.

Let $d'\geq 0$ such that $\nu_1'+d' \models \overline{C}$ and let $\nu_2'=(\nu_1'+d')[X:=0]$. Then let us show that there exists $d\geq 0$ such that $\nu_1 + d \models C$,
and
$(\ell_2,\nu_2)R_\epsilon (\ell'_2,\nu'_2)$ with $\nu_2=(\nu_1+ d)[X:=0]$. To prove that, one only needs to show that $\nu'_1+d'$ is in the "adherence" of the region of $\nu_1+ d $, and that $\nu_1+ d $ and $\nu'_1+  d'$ are $\epsilon$-neighbours.

Consider the interval
\[I=\enstq{ d\geq0}{\begin{array}{c}
		\nu_1+ d\models C\\ 
		\nu_1+ d\sqsubset r\text{ s.t. }\nu_1'+ d' \sqsubset \adh{r}\\
\end{array}}\]
Since $\Gcal_\prec$ is region trimmed, there exists $ d$ such that $\nu_1+ d\models C$, and every valuation of the form $\nu+ d$ that satisfies $C$ with $\nu\sqsubset \reg{\ell_1}$ belong to a unique region $r_t$. Since $\nu'_1\sqsubset \regadh{\ell_1}$ and $\nu'_1 +  d' \models \overline{C}$, $\nu'_1 +  d' \sqsubset \adh{r_t}$.
Therefore, $I$ is not empty.

If $ d' \in I$, then $(\nu_1+ d')[X:=0]$ and $\nu_1'+ d'$ are $\epsilon$-neighbours.
Otherwise, when $ d'\not\in I$, it entails that $I$ is constrained by some guards in $C$. These guards can be of the form $x=1$, $x>0$, $x<1$ or $x\leq 1$ for any $x\in\Xcal$
\footnote{There cannot be any guard of the form $x=0$ in $C$, otherwise $\Delta=\{ d'\}=\{0\}$. 
}
.

		\begin{itemize}
			\item If there is a guard $x=1$ in $C$ for some $x\in\Xcal$, then $I$ is the singleton $\{1-\nu_1(x)\}$. Indeed $I$ is not empty and $ d\eqdef 1-\nu_1(x)$ is the only delay that can satisfy the guard $x=1$. Note that, for the same reason, $ d'=1-\nu'_1(x)$. 
 By definition of $I$, $\nu_2'$ is in the adherence of the region that contains $\nu_2=(\nu_1+ d)$. Moreover, for all $y\in\Xcal$, if $y\in X$, then $\nu_2(y)=\nu_2'(y)=0$ and if $y\in\Xcal\setminus X$, then 
			$
			(\nu_1(y) +  d)-(\nu_1'(y) +  d' \in\intff{-(\epsilon_1+ d'- d)}{\epsilon_2-( d'- d)}
			$.
			Now, $ d'- d=\nu_1(x)-\nu_1'(x)\in\intff{-\epsilon_1}{\epsilon_2}$ 
			entails that $\epsilon_2'\eqdef \epsilon_2-( d'- d)\geq 0$, and $\epsilon_1'\eqdef \epsilon_1+( d'- d)\geq 0$. Finally, $\epsilon_1'+\epsilon_2'=\epsilon_1+\epsilon_2<\epsilon$.
			\end{itemize}
Now assume that there are no such guards in $C$:
			\begin{itemize}
		\item If there is a guard $x>0$ in $C$ for some $x\in\Xcal$
		such that $\nu_1+ d'\not\models (x>0)$ but $\nu_1' +  d'\models (x\geq 0)$ then $\nu_1(x)= d'=0$.
		Furthermore, since $\Gcal_\prec$ is trimmed, $\nu'_1(x)=0$. 
Therefore, $I$ is an interval of the form $]0,\dots]$.
Pick $ d\in I$ such that $0< d<\epsilon-(\epsilon_1+\epsilon_2)$.
Let
$\epsilon_2'\eqdef \epsilon_2-( d'- d)=\epsilon_2+ d$, and $\epsilon_1'\eqdef \max\{0,\epsilon_1+( d'- d)\}\geq0$. Then $\epsilon_1'+\epsilon_2'\leq \epsilon_1+\epsilon_2 +  d <\epsilon$, hence $\nu_1+ d$ and $\nu_1'+ d'$ are $\epsilon$-neighbours.	
\item If there is a guard $x< 1$ in $C$ (resp. $x\leq 1$) for some $x\in\Xcal$,
such that $\nu_1'+ d'\models (x\leq 1)$ but $\nu_1+ d' \not
\models (x< 1)$ (resp. $x\leq 1$), then $ d< d'$ for all $ d\in I$.
If $ d' + \nu_1'(x) - \nu_1(x)\in I$ then pick $ d=  d' + \nu_1'(x) - \nu_1(x)$.
			Since $ d'- d=\nu_1(x)-\nu_1'(x)\in\intff{-\epsilon_1}{\epsilon_2}$, $\nu_1+ d$ and $\nu_1'+ d'$ are $\epsilon$-neighbours.

Otherwise 
pick $ d\in I$ such that
$ 1 -\epsilon + (\epsilon_1+\epsilon_2) <\nu_1(x)+ d < 1$.
Then $\nu_1(x)- \nu_1'(x)< d'- d< \nu_1(x)- \nu_1'(x) + \epsilon-(\epsilon_1+\epsilon_2)$, with $\nu_1(x)- \nu_1'(x)
\in \intff{-\epsilon_1}{\epsilon_2}$.
Let
$\epsilon_2'\eqdef \max(0,\epsilon_2-( d'- d))\geq 0$ and $\epsilon_1'\eqdef \epsilon_1+( d'- d)\geq0$. Then $\epsilon_1'+\epsilon_2'<\epsilon$, hence $\nu_1' +  d'$ is an $\epsilon$-neighbour of $\nu_1+ d$.\qedhere
			
\end{itemize}

\end{proof}

We can now prove the main result of this section:
\begin{proof}[Proof of \cref{thm-sim-spicy}]

$\Gcal_\prec$ and $\Gcal_\preceq$ start from the same configuration, thus their first configurations are in bisimulation $R_\epsilon$ for any $\epsilon>0$.

Therefore, according to \cref{lem-bisim-WTG},
\[|\Value{\Gcal_\prec}-\Value{\Gcal_\preceq}|\leq \sup
\enstq{|\weight(\rho_\prec)-\weight(\rho_\preceq)|}{\begin{array}{c}
		\rho_\prec,\rho_\preceq \text{ finite plays of } \Gcal_\prec,\Gcal_\preceq\\ 
		\rho_\prec \,R_\epsilon\, \rho_\preceq\\
		\end{array}}
\;.\]

In a weightless kernel, for any finite run $\rho$ of length $n$, 
$\weight(\rho)=\wout(\rho^C(n))$. Let $s$ be the maximal slope (in absolute value, in one variable) in functions $\wout(\ell,\cdot)$. Then 

\begin{align*}
\sup&\setof{|\weight(\rho_1)-\weight(\rho_2)|}{\rho_1 \,R_\epsilon\, \rho_2}\\ 
&\leq \sup \setof{|\wout(\ell,\nu_1)-\wout(\ell,\nu_2)|}{\ell \in G, 
(\ell,\nu_1)\,R_\epsilon\, (\ell,\nu_2)}\\
&\leq \epsilon\cdot s \cdot |\Xcal|
\end{align*}
We conclude with \cref{lem-bisim-WTG}.

\end{proof}

\AP Note that $\Gcal_\preceq$  is not a "$\intfo01$-\wtg", but a "$\intff01$-\wtg", i.e., every accessible valuation is in $\intff01^\Xcal$.
Furthermore, $\Gcal_\preceq$ is not exactly a "region \wtg", even though relaxing preserve region assignment in some sense:


\begin{restatable}{definition}{relaxedreg}
\label{relaxedreg}
A \textbf{""relaxed""} "region \wtg" is a "$\intff01$-\wtg" without strict guards in which each location $\ell$ is assigned a "region" $\reg{\ell}$:
for any transition $t:\ell\xrightarrow{C,X} \ell'$, the valuations $\nu+ d$, with $\nu\sqsubset \regadh{\ell}$ and $ d\geq 0$, that satisfy $C$ are contained in $\adh{r}$ for a unique "region" $r$, such that $r[X:=0]=\reg{\ell'}$. Furthermore,
 initial configuration $(i,\nu)$ must verify $\nu\in\regadh{i}$.
A "relaxed" "region \wtg" is a \textbf{"trimmed"}
if for any transition $t:\ell\xrightarrow{C,X} \ell'$  and any "region" $r=\reg{\ell}$ in $\Gcal$,
	\begin{itemize}
		\item for any valuation $\nu\sqsubset \adh{r}$ there exists a $ d\geq0$ such that $\nu+ d\models C$.
		\item for any $c\in C$, there exists a valuation $\nu\sqsubset \adh{r}$ and some $ d\geq 0$ such that $\nu+ d$ is in $\intff01^\Xcal$ and
		$\nu+ d\not\models c$.
	\end{itemize}
\end{restatable}

\begin{lemma}
Let $\Gcal_\prec$ be a "trimmed" "region" "$\intfo01$-\wtg" of "region assignment" $\mathsf{reg}$. 
Let $\Gcal_\preceq$ be a copy of $\Gcal_\prec$ where 
\begin{itemize}
\item every "guard" has been relaxed, i.e., every "guard" of the form $x>0$ and $x<1$ have been replaced by $x\geq 0$ or $x\leq 1$ respectively, for all $x\in\Xcal$.
\item useless "guards" (see the second point of \cref{relaxedreg}) have been removed.
\end{itemize}

Then $\Gcal_\preceq$ is a "relaxed" "trimmed" "region \wtg", of "region assignment" $\mathsf{reg}$.
\end{lemma}

\subsubsection{Adding Resets to every Transition}
\label{sec-simpl}

\begin{lemma}
\label{lem-no-mixed-gadgets}
Let $A,B\in\{(\Minsf,\Maxsf),(\Maxsf,\Minsf)\}$.
In a relaxed weightless region kernel $\intff01$-\wtg, consider a transition $t:\ell_A \xrightarrow{C,X} \ell_B$ between a location $\ell_A$ belonging to Player $A$, and a location $\ell_B$ belonging to Player $B$, such that $C$ has no guard of the form $x=0,x=1,x<1$, and $X=\emptyset$. Then adding the guard $x=1$ to $C$ for some $x\in X^\uparrow_{\ell_A}$ does not change the value. 
\end{lemma}
\begin{proof}
Pick some $x\in X^\uparrow_{\ell_A}$.
By picking a delay $ d<1-\nu(x)$, Player $A$ offers Player $B$ more options than if they picked $ d=1-\nu(x)$. From the perspective of $B$, if a larger delay is advantageous, then they can take it from $\ell_B$ at cost $0$. Hence it is optimal for either $A$ or $B$ to pick the largest delay possible, i.e. $ d=1-\nu(x)$. However, since $w(\ell_A)=w(\ell_B)=0$, forcing $A$ to take a delay in $\ell_A$ which would have been taken in $\ell_B$ by Player $B$ given the chance does not change the value.
Thus, restricting $A$ to strategies which, when choosing $t$ from a valuation $\nu$, choose a delay $ d$ such that $(\nu+ d)(x)=1$ for all $x\in X^\uparrow_{\ell_A}$ does not change the value of the \wtg.
\end{proof}
\begin{restatable}{lemma}{allreset}
\label{thm-all-reset}
For any weightless "relaxed" "region" "trimmed"  "$\intff 01$-""sub-\wtg" $\Gcal$, such that $\Gcal$ has no requirement $x<1$ for any $x\in\Xcal$ then
there exists a weightless "relaxed" "region" "trimmed" "$\intff 01$-""sub-\wtg" $\Gcal'$ of same value and verifying the same conditions such that every transition of $\Gcal'$ is a reset transition or a transition to the target location.
Furthermore, any transition of $\Gcal'$ with, for some clock $x$, a "guard" of the form $x=0$ or $x=1$, resets $x$.
\end{restatable}
\begin{proof}
Let $\Gcal=(L_\Minsf, L_\Maxsf,G,\mathcal{X},T,w,\wout)$. 
We assume that $\Maxsf$ does not have full control over any cycle (\ie in any cycle there is a $\Minsf$ location from which she can decide to leave the cycle). Indeed, if $\Maxsf$ could reach such a cycle, the Value of $\Gcal$ would be $+\infty$.
Furthermore, let us assume that
there is no $\Minsf$ self-loop (a transition $\ell\in L_\Minsf\xrightarrow{C,X}\ell$) with $X=\emptyset$ in $\Gcal$: it does not make strategic sense for $\Minsf$ to take such a loop, hence it can be deleted without change in value.

Let us transform $\Gcal$ through the following operations.
For any $t:\ell_1 \xrightarrow{C,X} \ell_2$ of $\Gcal$ such that $\ell_2$ is not a target location, and $X=\emptyset$ and $C$ has no guards of the form $x=1$ for any clock $x$:
\begin{itemize}
\item If $C$ has an $x=0$ requirement for some clock $x$, then add $x$ to $X$.
\item If $C$ has no $x=0$ requirement for any clock $x$, and $\ell_1$ and $\ell_2$ belong to the same player, then remove $t$
and, for any $t':\ell_2\to \ell_3$
 create a transition $t'':\ell_1\to\ell_3$ such that $C(t'')=C(t)\cup C(t')$. If this create a self-loop without reset, remove it.

\item If $C$ has no $x=0$ requirement for any clock $x$, and $\ell_1$ and $\ell_2$ do not belong to the same player, then let us add an $x=1$ requirement to $C$ where $x\in X^\uparrow_{\ell_1}$. According to \cref{lem-no-mixed-gadgets}, it does not change the value.
\end{itemize}

After these transformations, every transition without reset in $\Gcal$ is either a transition to a target location, or has an $x=1$ guard for some $x\in\Xcal$.

Then let us build $\Gcal'=(L'_\Minsf, L'_\Maxsf,G',\mathcal{X},T',w',\wout')$ a kernel \wtg
where:
\begin{itemize}
\item For any player $\Player$, let $L'_\Player=L_\Player \cup \setof{\ell_\downarrow}{\ell \in L_\Player}$.
\item Let $G'= G \cup \setof{\ell_\downarrow}{\ell \in G}$.
\item Start from $T'=\emptyset$.
For any $t:\ell\xrightarrow{C,X}\ell'$ in $T$, 
	\begin{itemize}
	\item if $\ell'\in G$ then for all $X\subseteq \upclock{\ell}$ add $t':\ell\xrightarrow{C,X} \ell'$ 
and	$t'':\ell_\downarrow\xrightarrow{C_\downarrow,X} \ell'_\downarrow$ 
	to $T'$, where $C_\downarrow
	\eqdef C\cup \setof{x=0}{x\in \upclock{\ell}}\setminus  \setof{x=1}{x\in \upclock{\ell}}$. Note that here $\upclock{\ell}$ is defined according to $\reg{\ell}$ the region assignment in the relaxed trimmed region \wtg $\Gcal$.
	\item if $\ell'\not\in G$ and $X=\emptyset$ then $(x=1)\in C$ for some $x\in \upclock{\ell}$.
	Then add $t':\ell \xrightarrow{C,\upclock{\ell}} \ell'_\downarrow$ 
	and $t'':\ell_\downarrow\xrightarrow{C_\downarrow,\upclock{\ell}} \ell'_\downarrow$ to $T'$.
	\item if $\ell'\not\in G$ and $X= \upclock{\ell}$, then
	add $t':\ell \xrightarrow{C,X} \ell'$  
	and $t'':\ell_\downarrow \xrightarrow{C_\downarrow,X} \ell'$
	to $T'$. 
	\item if $\ell'\not\in G$ and $X\neq\emptyset$, then
	add $t':\ell \xrightarrow{C,X} \ell'$  
	and $t'':\ell_\downarrow \xrightarrow{C',X'} \ell'_\downarrow$
	to $T'$.
	\end{itemize}
\item For all $\ell\in G$, let $\wout'(\ell,\nu)=\wout(\ell,\nu)$, $\wout'(\ell_\downarrow,\nu)=\wout(\ell,\nu')$
where $\nu'(x)=\nu(x)$ for all $x\not\in X$, and $\nu'(x)=1$ otherwise.

\end{itemize}

Intuitively, in $\Gcal$, when one or several clocks are made to reach $1$ by a guard, there will usually be some urgent transitions taken until all these clocks have been reset. In $\Gcal'$, those clocks are immediately reset, and control moves to a location $\ell_\downarrow$. All paths leaving $\ell_\downarrow$  have $(x=0)$ conditions (for all $x$ that have been reset in $\Gcal'$ but not in $\Gcal$) to guarantee urgency. 
A valuation $\nu'$ in a location $\ell_\downarrow$ in $\Gcal'$ is thus equivalent to a valuation $\nu$ in $\ell$ in $\Gcal$ iff $\nu(x)=\nu'(x)$ if $x\not\in \upclock{\ell}$, and $\nu(x)=1$ and $\nu'(x)=0$ for $x\in \upclock{\ell}$.

$\Gcal'$ is a relaxed region \wtg with region assignment $\mathsf{reg}'(\ell)=\reg{\ell}$ and $\mathsf{reg}'(\ell_\downarrow)=(X_0\cup X_p, X_1,\dots,X_{p-1})$ when $\reg{\ell}=(X_0, X_1,\dots,X_{p})$.
From there, trim $\Gcal'$ to obtain a relaxed trimmed region \wtg.

\end{proof}

\subsection{Value Iteration in Two-clock Kernel Games}
\label{sec-value-kernel}

In this section, we prove \cref{lem-linear-kernel}
using the value iteration paradigm (see \cite{alur2004,BGHM16,brihaye2022}):

\begin{definition}[Value iteration algorithm]
Let $\Gcal=(L_{\Minsf},L_{\Maxsf},G,\Xcal,T,w)$ be a "region" (sub)-\wtg.
The \textbf{""value iteration algorithm""} builds, for each location $\ell$ and for all $k\geq 0$, a function $\intro*\opt{\ell}{k}: \reg{\ell} \to \Rbb\cup\{+\infty\}$ such that
$\opt{\ell}{k}(\nu)$ is the value of the game starting in $\ell$ with clock valuation $\nu$, when Min is restricted to strategies reaching G within k steps.
The $\opt{}{}$ functions are built inductively:
\begin{itemize}
\item If $\ell \in G$, $\opt{\ell}{k}$ is the constant $0$ function for any $k\in \mathbb{N}$, 
(or $\wout(\ell,\cdot)$ if 
$\Gcal$ is a "sub-\wtg").
\item If $\ell\not\in G$, $\opt{\ell}{0}$ is the constant $+\infty$ function.
\item for any $k\in \Nat$, $\opt{\ell}{k+1}$ is obtained from the $\opt{}{}$ functions at step $k$: if $\ell$ belongs to $\Minsf$ (resp.~$\Maxsf$), then
\[\opt{\ell}{k+1}(\nu)=\underset{\text{(resp.~$\sup$) }}{\inf}  \enstq{\opt{\ell'}{k}((\nu+\delta)[X:=0])+\delta \cdot w(\ell)+w(t)}
{\begin{array}{c}
t:\ell \overset{C,X}{\rightarrow} \ell'\in T,\\
 \delta\geq 0, \nu+\delta \models C
\end{array}
}.\]
\end{itemize}

\end{definition}

Note that $\opt{\ell}{k}(\nu)=\Value{\Gcal}^{\leq k}(\ell,\nu)$ the value of the game starting from configuration $(\ell,\nu)$ when $\Minsf$ must reach $G$ in at most $k$ transitions.
Naturally, the value iteration operator is monotonous ($\opt{\ell}{k+1}(\nu)\leq \opt{\ell}{k}(\nu)$ for all valuation $\nu$). If there exists $k$ such that, for all locations $\ell$, $\opt{\ell}{k+1}= \opt{\ell}{k}$, then the "value iteration algorithm" terminates and $\opt{\ell}{k}(\nu)=\Value{\Gcal}(\ell,\nu)$. Consequently, the value of the \wtg is obtained even when considering plays of length at most $k$.

Unfortunately, termination is not guaranteed (otherwise the "Value Problem" would be decidable), except in the case of
one-clock \wtgs with non-negative weights (\cite{bouyer2006}).
 
Here is an example where the "value iteration algorithm" does not terminate.

\begin{figure}
\centering
\includegraphics[scale=1]{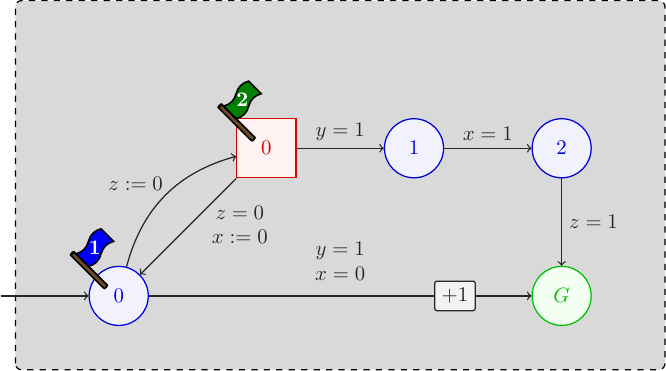}
\caption{ A "\wtg " where the "value iteration algorithm" does not terminate.
Blue circle locations belong to $\Minsf$, red square locations belong to $\Maxsf$, the green circle location $G$ is the target. The flags serve to easily refer to locations. The $+1$ label is a transition cost.}
\label{fig-3clock}
\end{figure}

\begin{example}
The "\wtg " $\Gcal$ in \cref{fig-3clock} is a $3$-clock, "almost non-Zeno \wtg" with a value of $1$. The "kernel" of $\Gcal$ contains only the cycle between \includegraphics[scale=.5]{standalones/flags/flag1.pdf} and \includegraphics[scale=.5]{standalones/flags/flag2.pdf}.
The cost of the "output edge" in \includegraphics[scale=.5]{standalones/flags/flag2.pdf} is $a+2\delta$ for a valuation $(x,y,z)=(\delta, a+\delta,0)$, whereas the cost of the "output edge" in \includegraphics[scale=.5]{standalones/flags/flag1.pdf} is exactly $1$, for a valuation $(x,y,z)=(0, 1,0)$.

An optimal strategy for $\Minsf$ is to loop in the "kernel" an arbitrary number of times: In this strategy, each time she enters \includegraphics[scale=.5]{standalones/flags/flag1.pdf} with valuation $(x,y,z)=(0,a,z)$, she should waits a delay $d$ such that $2d = 1-a$ and enters \includegraphics[scale=.5]{standalones/flags/flag2.pdf}, then $\Maxsf$ can either reach $G$ with a cost of exactly $a+2d=1$, or return to \includegraphics[scale=.5]{standalones/flags/flag1.pdf}. When $\Minsf$ decides to end the game, she then picks $d=1-a$ instead. Then $\Maxsf$ chooses between going to $G$ at cost $1 + (1-a)$, or letting $\Minsf$ leave from \includegraphics[scale=.5]{standalones/flags/flag1.pdf} at cost $1$.
 Each time $\Minsf$ takes the cycle with delay $d$ such that $2d= 1-a$, $y$ gets closer to $1$, thus minimizing the cost of picking 
 $d=1-a$ at some point. Thus $\Minsf$ has a strategy to reach $G$ with cost $>1$, but arbitrarily close to $1$ depending on how long she plays.
This entails that the "value" of this \wtg is obtained by considering arbitrarily long plays. Thus, the "value iteration algorithm" does not terminate on this example.
\end{example}

However, the "value iteration algorithm" terminates for two-clock "kernel \wtgs".

\thmKernel*

\begin{proof}
Without loss of generality, let us assume that $\Gcal$ is a weightless "relaxed" "region" "trimmed" "$\intff01$-""sub-\wtg" where every transition is either a transition to a target location, or resets at least one clock (see \cref{lem-zero-one,thm-sim-spicy,thm-all-reset}). Furthermore, we assume that there is no transition to the initial location $i$.\footnote{This can be done by making a copy of the initial location, such that every transition that should enter the initial location only enter the copy instead.}

For every location $\ell\not\in G \cup \{i\}$, the set of valuations $\regadh{\ell}$  is either $\setof{(0,y)}{y\in \intff01}$ or $\setof{(x,0)}{x\in \intff01}$, or the singleton $\{(0,0)\}$. Note that the $\opt{\ell}{}$ functions will be defined on $ \regadh{\ell}$ for any location $\ell$. This entails that the "value iteration algorithm" will mainly build $1$-dimensional functions. 

\AP Let us highlight this observation with a subtle change of variable. Remember that  the "circular clock difference" $\DDelta$ of a valuation $\nu\in\regadh{\ell}$ is defined as: 
\[\DDelta(\nu)=\begin{cases}
y \text{ when } \nu=(0,y)\text{ with $y\geq 0$,}\\
1-x \text{ when } \nu=(x,0) \text{ with $x>0$.}
\end{cases}\]

\AP Now, for any  $\ell\not\in G \cup \{i\}$, for any $k\in\Nat$, let
$\intro*\Opt{\ell}{k}(\DDelta(\nu))\eqdef\opt{\ell}{k}(\nu)$ for all $\nu\in \regadh{\ell}$. The function $\Opt{\ell}k$ is either defined on $\{0\}$ if $\regadh{\ell}=\{(0,0)\}$, or on $\intff01$.


\begin{figure}[t]
\centering
\scalebox{.8}{
\begin{tikzpicture}

\def\xx{5}
\def\dd{2}
\def\xdd{3}
\def\ee{4}
\def\nn{.1}
\def\zz{7}

\fill[gray!10] (0, 0) -- (0,\xx) -- (\xx,\xx) -- (0,0);
\draw[gray, very thin] (0, 0) -- (\xx,0)-- (\xx, \xx);
\draw[->,gray, very thin] (\xx, 0) -- (\xx+.2, 0);
\draw[->,gray, very thin] (0, \xx) -- (0,\xx+.2);
\draw[dashed,gray, very thin] (-.6, \dd) -- (-1.5,\dd);
\draw[dashed,gray, very thin] (-.2, \dd) -- (0,\dd);
\draw[thick] (0, 0) -- (0,\xx) -- (\xx,\xx) -- (0,0);
\node[gray] at (\xx+.4,0){$x$};
\node[gray] at (0,\xx+.4){$y$};

\node (A) at (0,\dd){};
\node (B) at (\ee-\dd,\ee){};
\node (C) at (0, \ee){};

\draw [blue, very thick] (0,\dd) -- (\ee-\dd,\ee);
\draw [blue, very thick,dashed] (\ee-\dd,\ee) --(0,\ee) ;
\draw [gray, very thin] (\ee-\dd,\ee) -- (\xx-\dd,\xx);

  \foreach \y in {.2,.4,...,1.8, 2.2, 2.4, ..., \xdd}{

  \draw [gray,thin] (\y,\y+\dd) --(0,\y+\dd);

  }
  
  \draw[->] (-1.6, 0) -- (-1.6, \xx+.2);

\node at (-.4,\dd){$\nu$};
\node at (\ee-\dd+.7,\ee){$\nu+\delta$};
\node at (-.7,\ee+.2){$\nu+\delta$};
\node at (-.7,\ee-.2){$[x:=0]$};

\draw (-1.7,\dd) -- (-1.5,\dd);
\draw (-1.7,\ee) -- (-1.5,\ee);
\draw (-1.7,0) -- (-1.5,0);
\node at (-2,\dd){$\Delta$};
\node at (-2.3,\ee){$\Delta+\delta$};

\fill[gray!10] (0+\zz, 0) -- (\zz,\xx) -- (\xx+\zz,\xx) -- (0+\zz,0);
\draw[gray, very thin] (0+\zz, 0) -- (\xx+\zz,0)-- (\xx+\zz, \xx);
\draw[->,gray, very thin] (\xx+\zz, 0) -- (\xx+.2+\zz, 0);
\draw[->,gray, very thin] (0+\zz, \xx) -- (0+\zz,\xx+.2);
\draw[thick] (0+\zz, 0) -- (0+\zz,\xx) -- (\xx+\zz,\xx) -- (0+\zz,0);
\node[gray] at (\xx+.4+\zz,0){$x$};
\node[gray] at (0+\zz,\xx+.4){$y$};

\node (A) at (0+\zz,\dd){};
\node (B) at (\ee-\dd+\zz,\ee){};
\node (C) at (0+\zz, \ee){};

\draw [blue, very thick] (0+\zz,\dd) -- (\ee-\dd+\zz,\ee);
\draw [blue, very thick,dashed] (\ee-\dd+\zz,\ee) --(\ee-\dd+\zz,0) ;
\draw [gray, very thin] (\ee-\dd+\zz,\ee) -- (\xx-\dd+\zz,\xx);

  \foreach \y in {.2,.4,...,1.8, 2.2, 2.4, ..., \xdd}{

  \draw [gray,thin] (\y+\zz,\y+\dd) --(\y+\zz,0);

  }
  
\draw[->] (\xx+\zz, -1.2) -- (0+\zz, -1.2);

\node at (-.4+\zz,\dd){$\nu$};
\node at (\ee-\dd-.7+\zz,\ee){$\nu+\delta$};
\node at (\ee-\dd+\zz,-.2){$\nu+\delta$};
\node at (\ee-\dd+\zz,-.7){$[y:=0]$};
\draw[dashed,gray, very thin] (\xx+\zz-\dd, 0) -- (\xx+\zz-\dd,-1.5);
\draw[dashed,gray, very thin] (\ee+\zz-\dd, -.8) -- (\ee+\zz-\dd,-1.1);
\draw[black] (\xx+\zz-\dd,-1.1) -- (\xx+\zz-\dd,-1.3);
\draw[black] (\ee+\zz-\dd,-1.1) -- (\ee+\zz-\dd,-1.3);
\draw[black] (\xx+\zz,-1.1) -- (\xx+\zz,-1.3);

\node at (\xx+\zz-\dd,-1.6){$\Delta$};
\node at (\ee+\zz-\dd+.1,-1.6){$1-\delta>$};

\end{tikzpicture}
}
\caption{Evolution of $\Delta$ in a one-clock-reset transition, without guards of the form $x=0$ or $x=1$, from region $\setof{(0,y)}{y\in \intff01}$.}
\label{fig-rry}

\end{figure}

\begin{figure}[t]
\centering
\scalebox{.8}{
\begin{tikzpicture}

\def\xx{5}
\def\dd{2}
\def\xdd{3}
\def\ee{4}
\def\nn{.1}
\def\zz{9}

\draw[gray, very thin] (0, 0) -- (0,\xx)-- (\xx, \xx);
\draw[->,gray, very thin] (\xx, 0) -- (\xx+.2, 0);
\draw[->,gray, very thin] (0, \xx) -- (0,\xx+.2);
\fill[gray!10] (0, 0) -- (\xx,0) -- (\xx,\xx) -- (0,0);
\draw[thick] (0, 0) -- (\xx,0) -- (\xx,\xx) -- (0,0);
\node[gray] at (\xx+.4,0){$x$};
\node[gray] at (0,\xx+.4){$y$};

\draw [blue, very thick] (\dd,0) -- (\ee,\ee-\dd);
\draw [blue, very thick,dashed] (\ee, \ee-\dd) --(\ee, 0);
\draw [gray, very thin] (\ee,\ee-\dd) -- (\xx,\xx-\dd);

  \foreach \y in {.2,.4,...,1.8, 2.2, 2.4, ..., \xdd}{

  \draw [gray,thin] (\y+\dd,\y) --(\y+\dd,0);

  }
  
  \draw[->]  (\xx,-1.2) -- (0,-1.2);

\node at (\dd,-.4){$\nu$};
\node at (\ee,\ee-\dd+.7){$\nu+\delta$};
\node at (\ee, -.3){$\nu+\delta$};
\node at (\ee, -.7){$[y:=0]$};
\draw[dashed,gray, very thin] (\dd,-.6) -- (\dd,-1.1);
\draw[dashed,gray, very thin] (\dd,-.3) -- (\dd,0);
\draw[dashed,gray, very thin] (\ee,-.9) -- (\ee,-1.1);

\draw (\dd, -1.1) -- (\dd, -1.3);
\draw (\ee,-1.1) -- (\ee,-1.3);
\draw (\xx,-1.1) -- (\xx,-1.3);
\node at (\dd, -1.5){$\Delta$};
\node at (\ee, -1.5){$\Delta-\delta$};

\fill[gray!10] (0+\zz, 0) -- (\xx+\zz,0) -- (\xx+\zz,\xx) -- (0+\zz,0);
\draw[gray, very thin] (0+\zz, 0) -- (\zz,\xx)-- (\xx+\zz, \xx);
\draw[->,gray, very thin] (\xx+\zz, 0) -- (\xx+.2+\zz, 0);
\draw[->,gray, very thin] (0+\zz, \xx) -- (0+\zz,\xx+.2);
\draw[thick] (0+\zz, 0) -- (\xx+\zz,0) -- (\xx+\zz,\xx) -- (0+\zz,0);
\node[gray] at (\xx+.4+\zz,0){$x$};
\node[gray] at (0+\zz,\xx+.4){$y$};

\draw [blue, very thick] (\dd+\zz,0) -- (\ee+\zz, \ee-\dd);
\draw [blue, very thick,dashed] (\ee+\zz, \ee-\dd) --(0+\zz, \ee-\dd) ;
\draw [gray, very thin] (\ee+\zz, \ee-\dd) -- (\xx+\zz, \xx-\dd);

  \foreach \y in {.2,.4,...,1.8, 2.2, 2.4, ..., \xdd}{

  \draw [gray,thin] (\y+\zz+\dd,\y) --(0+\zz,\y);

  }
  
\draw[<-] (-1.6+\zz,\xx) -- (-1.6+\zz, 0);

\node at (\dd+\zz,-.4){$\nu$};
\node at (\ee+\zz,\ee-\dd-.7){$\nu+\delta$};
\node at (-.7+\zz,\ee-\dd+.2){$\nu+\delta$};
\node at (-.7+\zz,\ee-\dd-.2){$[x:=0]$};
\draw[dashed,gray, very thin] (\zz,\xx-\dd) -- (-1.5+\zz,\xx-\dd);

\draw[black] (-1.5+\zz,\xx-\dd) -- (-1.7+\zz,\xx-\dd);
\draw[black] (-1.5+\zz,\ee-\dd) -- (-1.7+\zz,\ee-\dd);
\draw[black] (-1.5+\zz,0) -- (-1.7+\zz,0);

\node at (-1.9+\zz,\xx-\dd){$\Delta$};
\node at (-2.3+\zz,\ee-\dd){$\delta<\Delta$};

\end{tikzpicture}
}
\caption{Evolution of $\Delta$ in a one-clock-reset transition, without guards of the form $x=0$ or $x=1$, from region $\setof{(x,0)}{x\in \intff01}$.}
\label{fig-rrx}
\end{figure}

The following observation motivates this change of variable:
In \cref{fig-rry,fig-rrx}, we consider a transition $t$, from locations $\ell$ to $\ell'$ such that the $\Opt{\ell}{}$ and $\Opt{\ell'}{}$ functions are defined on $\intff01$, with no $=0$ or $=1$ guards.
For a transition $(\ell,\nu) \overset{d,t}\rightarrow
(\ell',\nu')$, we then observe that $\DDelta(\nu')\in [\DDelta(\nu), 1]$ if $\nu(x)=0$ (\cref{fig-rry}), and $\DDelta(\nu')\in [0,\DDelta(\nu)]$ if $\nu(y)=0$ (\cref{fig-rrx}). Thus the variation of $\DDelta$ only depends on the "region" of $\ell$, not $\ell'$. In other words, it does not matter which clocks are reset by a transition, just the number of clocks that are reset.

Let us now define the induction relation between the $\Opt{}{}$ functions.

\AP Let $\ell$ be a location such that $\ell\not\in G\cup\{i\}$. Moreover, assume that $\ell$ belongs to $\Minsf$ (resp.~ $\Maxsf$ ).
For $k=0$, $\Opt{\ell}{k}=\Delta \mapsto +\infty$.
Let $\intro*\Next(\ell)$ be the set of outgoing transitions from a location $\ell$ belonging to $\Minsf$ (resp.~ $\Maxsf$ ), with $\ell\neq i$. For each $t:\ell\overset{C,X}{\rightarrow} \ell'$, for any $k\in\Nat$, let us define a function $\Opt{t}{k}$ such that $\Opt{\ell}{k+1}=\min_{t\in\Next(\ell)} \Opt{t}{k}$ (resp.~ $\max$):

\begin{itemize}
\item if $\ell'$ is a goal location, then for all $k\in\Nat$,
	\begin{itemize}
	\item If $\regadh{\ell}=\setof{(0,y)}{y\in \intff01}$,
	then 
	\[\Opt{t}{k}=\Opt{t}{0}: \Delta \mapsto\underset{(d,\Delta+d)\models C}{\underset{0\leq d\leq 1-\Delta} {\inf}}\, \wout(\ell',(d,\Delta+d)[X:=0])\text{ (resp.~ $\sup$).}\]
	\item If $\regadh{\ell}=\setof{(x,0)}{x\in \intff01}$,
		then \[\Opt{t}{k}=\Opt{t}{0}: \Delta \mapsto\underset{(1-\Delta+d,d)\models C}{\underset{0\leq d\leq \Delta} {\inf}}\, \wout(\ell',(1-\Delta+d,d)[X:=0])\text{ (resp.~ $\sup$).}\]
	\item If $\regadh{\ell}=\{(0,0)\}$, 
	then $\Opt{t}{k}(0)=\Opt{t}{0}(0)= \underset{(d,d)\models C}{\underset{0\leq d\leq 1} {\inf}}\, \wout(\ell',(d,d)[X:=0])$ (resp.~ $\sup$).
	
	\end{itemize}
	
\item if $\ell'$ is not a goal location, then $X\neq\emptyset$ ($t$ resets one or two clocks). Thus $\Opt{\ell'}{k}$ is either defined on the singleton set, or $\intff01$.
	\begin{itemize}
	\item If $\Opt{\ell'}{k}$ is defined on the singleton set, then since $\Gcal$ is a "relaxed" "trimmed" "region \wtg", every valuation in $\regadh{\ell}$ can reach $(\ell',(0,0))$. Therefore $\Opt{t}{k}= \Delta \mapsto \Opt{\ell'}{k}(0)$.
	\item Otherwise, $\Opt{\ell'}{k}$ is defined on $\intff01$. Then:
		\begin{itemize}
		\item If $\regadh{\ell}=\setof{(0,y)}{y\in \intff01}$,
			\begin{itemize}
			\item If $(x=0)\in C$ (thus $X=\{x\}$) or $(y=1)\in C$ (thus $X=\{y\}$),\footnote{Note that since $\Gcal$ is a "relaxed" "trimmed" "region \wtg", both guards cannot be in $C$ at the same time. 
} 
			then this forces a delay such that $t$ preserves $\DDelta$. 
			\item Otherwise, as observed in \cref{fig-rry},
			\[\Opt{t}{k}: \Delta \mapsto\underset{\Delta\leq \Delta'\leq 1}{\inf}\, \Opt{\ell'}{k}(\Delta')\text{ (resp.~ $\sup$).}\]
			\end{itemize}
		\item Symmetrically, if $\regadh{\ell}=\setof{(x,0)}{x\in \intff01}$,
		\begin{itemize}
			\item If $(y=0)\in C$ or $(x=1)\in C$, then $\Opt{t}{k}= \Delta \mapsto \Opt{\ell'}{k}(\Delta)$.
			\item Otherwise, as observed in \cref{fig-rrx},
			\[\Opt{t}{k}: \Delta \mapsto\underset{0\leq \Delta'\leq \Delta}{\inf}\, \Opt{\ell'}{k}(\Delta')\text{ (resp.~ $\sup$).}\]
			\end{itemize}
		\item If $\regadh{\ell}=\{(0,0)\}$,
		$\Opt{t}{k}(0)= \underset{0\leq \Delta\leq 1}\inf \Opt{\ell'}{k}(\Delta)$ (resp.~ $\sup$).
		\end{itemize}
	\end{itemize}

\end{itemize}

We call the functions $\Opt{t}{0}$, for all transitions $t$ to a goal location, the \textbf{""projected output functions""} of $\Gcal$.  

The "projected output functions" serve to initialize the "value iteration algorithm" on the $\Opt{}{}$ functions, as do the $\wout(\ell,\cdot)$ functions for the value iteration algorithm on the $\opt{}{}$ functions. Observe that, since $\wout(\ell,\cdot)$ is piecewise linear and continuous for any $\ell\in G$, the "projected output functions" are continuous piecewise linear functions.

\begin{figure}[ht]
\includegraphics[scale=.7]{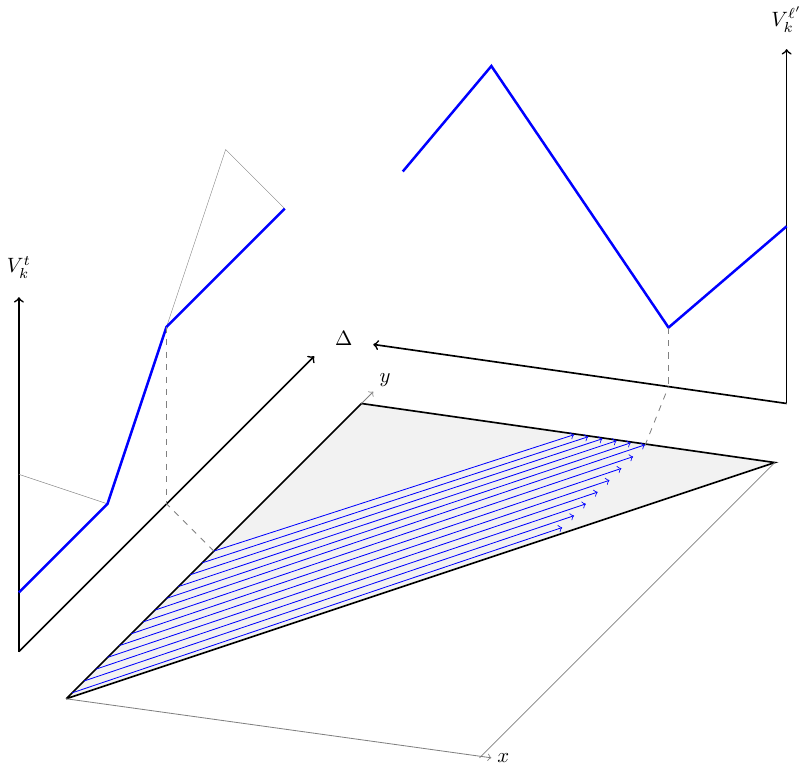}
\caption{$\Opt{}{k}$ induction relation, with $t$ a transition from a $\Minsf$ location $\ell$ where $x\leq y$. $t$ has no $=0$ or $=1$ guards, and applies $y:=0$.}
\label{fig-slash}
\end{figure}

Thus the functions $\Opt{\ell}{k}$ are by construction continuous piecewise linear functions. Furthermore,
as can be seen in \cref{fig-slash}, $\Opt{t}{k}$ is obtained from $\Opt{\ell'}{k}$ by replacing $\Opt{\ell'}{k}$ on a finite number of intervals by constant functions $\Delta\mapsto c$ (while preserving continuity) where the constants $c$ are taken among local extrema of $\Opt{\ell'}{k}$.

Hence every linear piece of a function $\Opt{\ell}{k}$ is:
\begin{itemize}
\item either equal to some "projected output function",
\item or of slope $0$, equal to some $z\in \mathbb{R}^+$, where $z$ is a local extremum of some function $\Opt{\ell}{k'}$ for $k'<k$. Hence, by induction, $z$ is either a local extremum of some "projected output function", or a local extremum of the minimum or maximum of two "projected output functions".
\end{itemize}

There is a finite number of such pieces, hence only a finite number of ways they can be assembled to make a \emph{continuous} piecewise linear function on $\intff01$.
Thus there is a finite number of functions of this form.

Since the $\opt{}{}$ and $\Opt{}{}$ functions  decrease at each iteration ($\Opt{\ell}{k+1}(\Delta)\leq\Opt{\ell}{k}(\Delta)$), the "value iteration algorithm" on the $\Opt{}{}$ functions terminates.

Furthermore, since no transition enters $i$, $\opt{i}{k}$ does not affect $\opt{\ell}{k+1}$ for any location $\ell$. Therefore, the "value iteration algorithm" on the $\opt{}{}$ functions aside from $\opt{i}{}$  stabilizes. Thus, if it terminates in $k$ steps for every location except $i$, then, adding $\opt{i}{}$, the "value iteration algorithm" terminates in at most $k+1$ steps.
\end{proof}

Note that all transformations in \cref{sec-interlude,sec-wlog,sec-all-simpl} only serve to make the termination argument of \cref{lem-linear-kernel} more visible. However, the "value iteration algorithm" on two-clock "kernel \wtgs" terminates even without these simplifications. Indeed, termination in $k$ steps entails that $\Minsf$ needs only to consider strategies that access a goal location in at most $k$ steps. Since the transformations described in \cref{sec-interlude,sec-wlog,sec-all-simpl} do not make arbitrarily long paths equivalent to one shorter path, then termination of the "value iteration algorithm" on the transformed "\wtgs" immediately implies termination on the original "\wtgs". This in turn entails that the "value iteration algorithm" terminates on the "semi-unfolding" $\tilde{\Gcal}$, thus on any "almost divergent" two-clock "\wtg ". 

Furthermore, the step-equivalence ratio between the original and transformed version is linear in the size of the game and the guards and weight constants, thus the transformations can be ignored in the following complexity analysis.

\paragraph*{Complexity analysis.}

Assume that the "output weight functions" of a "kernel game" $\Gcal$ consist of a total of $k$ pieces. Then the number of possible $\Opt{}{}$ functions,
which bounds the number of steps of the "value iteration algorithm" according to \cref{lem-linear-kernel},
 is exponential in $k$.

In the "semi-unfolding" of a weighted timed game $\Gcal$, the number of pieces of the piecewise linear value functions 
$W_n$ of nodes $n$
 increases at most exponentially along a branch. Each branch has length bounded by
$|\RC{\Gcal}|\cdot \gamma$ (see \cite{Busatto2018}). 
Since $\gamma$ is polynomial in $|\RC{\Gcal}|$, applying the "value iteration algorithm" requires double exponential time in 
$|\RC{\Gcal}|$.

Finally, the size of $\TC{\Gcal}$ is bounded by $|\RC{\Gcal}|^{|\RC{\Gcal}|\cdot (\gamma+1)}$, hence computing the "value" is double-exponential in $|\RC{\Gcal}|$ (or triple-exponential in the size of the input, where weights and guards are given in binary).

Note that the approximation algorithms for "almost non-Zeno \wtgs " and "almost divergent \wtgs" given in \cite{Bouyer2015,Busatto2018} are exponential in $|\RC{\Gcal}|$. Hence computing the exact value of a two-clock "almost divergent \wtg " is exponentially more complex than computing an approximation.

\clearpage

\bibliography{biblio} 

@book{Minsky-Computations,
	author = {Minsky, Marvin L.},
	publisher = {Prentice-Hall},
	series = {Prentice-Hall Series in Automatic Computation},
	title = {Computation: Finite and Infinite Machines},
	pages = {255-258},
	year = 1967
}

@inproceedings{alur2004,
  author = {Alur, Rajeev and Bernadsky, Mikhail and Madhusudan, P.},
  editor = {D{\'i}az, Josep and Karhum{\"a}ki, Juhani and Lepist{\"o}, Arto and Sannella, Donald},
  title = {Optimal Reachability for Weighted Timed Games},
  booktitle = {Automata, Languages and Programming},
  year = {2004},
  publisher = {Springer Berlin Heidelberg},
  address = {Berlin, Heidelberg},
  pages = {122--133},
  abstract = {Weighted timed automata are timed automata annotated with costs on locations and transitions. The optimal game-reachability problem for these automata is to find the best-cost strategy of supplying the inputs so as to ensure reachability of a target set within a specified number of iterations. The only known complexity bound for this problem is a doubly-exponential upper bound. We establish a singly-exponential upper bound and show that there exist automata with exponentially many states in a single region with pair-wise distinct optimal strategies.},
  isbn = {978-3-540-27836-8},
  timestamp = {Tue, 23 May 2017 01:00:00 +0200},
  biburl = {https://dblp.org/rec/conf/icalp/AlurBM04.bib},
  bibsource = {dblp computer science bibliography, https://dblp.org},
  doi = {10.1007/978-3-540-27836-8_13},
  _bib2doi_selected = {dblp:/rec/conf/icalp/AlurBM04.bib},
  _bib2doi_confirmed = {true},
}

@article{brihaye2022,
  title = {One-Clock Priced Timed Games with Negative Weights},
  volume = {Volume 18, Issue 3},
  issn = {1860-5974},
  url = {http://dx.doi.org/10.46298/lmcs-18(3:17)2022},
  doi = {10.46298/lmcs-18(3:17)2022},
  journal = {Logical Methods in Computer Science},
  publisher = {Centre pour la Communication Scientifique Directe (CCSD)},
  author = {Brihaye, Thomas and Geeraerts, Gilles and Haddad, Axel and Lefaucheux, Engel and Monmege, Benjamin},
  year = {2022},
  month = {aug},
  timestamp = {Tue, 24 Jan 2023 00:00:00 +0100},
  biburl = {https://dblp.org/rec/journals/lmcs/BrihayeGHLM22.bib},
  bibsource = {dblp computer science bibliography, https://dblp.org},
  _bib2doi_selected = {dblp:/rec/journals/lmcs/BrihayeGHLM22.bib},
  _bib2doi_confirmed = {true},
}

@inproceedings{BCFL-fsttcs04,
  author = {Patricia Bouyer and Franck Cassez and Emmanuel Fleury and Kim Guldstrand Larsen},
  editor = {Kamal Lodaya and Meena Mahajan},
  title = {Optimal Strategies in Priced Timed Game Automata},
  booktitle = {{FSTTCS} 2004: Foundations of Software Technology and Theoretical Computer Science, 24th International Conference, Chennai, India, December 16-18, 2004, Proceedings},
  year = {2004},
  publisher = {Springer},
  address = {Berlin, Heidelberg},
  pages = {148--160},
  abstract = {Priced timed (game) automata extend timed (game) automata with costs on both locations and transitions. In this paper we focus on reachability priced timed game automata and prove that the optimal cost for winning such a game is computable under conditions concerning the non-zenoness of cost. Under stronger conditions (strictness of constraints) we prove that in case an optimal strategy exists, we can compute a state-based winning optimal strategy.},
  isbn = {978-3-540-30538-5},
  timestamp = {Wed, 24 May 2017 01:00:00 +0200},
  biburl = {https://dblp.org/rec/conf/fsttcs/BouyerCFL04.bib},
  bibsource = {dblp computer science bibliography, https://dblp.org},
  doi = {10.1007/978-3-540-30538-5_13},
  volume = {3328},
  url = {https://doi.org/10.1007/978-3-540-30538-5\_13},
  series = {Lecture Notes in Computer Science},
  _bib2doi_selected = {dblp:/rec/conf/fsttcs/BouyerCFL04.bib},
  _bib2doi_confirmed = {true},
  _bib2doi_finished = {true},
}

@article{BusattoGastonMR23,
  author = {Damien Busatto{-}Gaston and Benjamin Monmege and Pierre{-}Alain Reynier},
  title = {Optimal controller synthesis for timed systems},
  journal = {Log. Methods Comput. Sci.},
  volume = {19},
  number = {1},
  year = {2023},
  timestamp = {Tue, 07 May 2024 01:00:00 +0200},
  biburl = {https://dblp.org/rec/journals/lmcs/BusattoGastonMR23.bib},
  bibsource = {dblp computer science bibliography, https://dblp.org},
  doi = {10.46298/lmcs-19(1:20)2023},
  _bib2doi_selected = {dblp:/rec/journals/lmcs/BusattoGastonMR23.bib},
  _bib2doi_confirmed = {true},
}

@inproceedings{behrmann2001,
  author = {Gerd Behrmann and Ansgar Fehnker and Thomas Hune and Kim Guldstrand Larsen and Paul Pettersson and Judi Romijn and Frits W. Vaandrager},
  editor = {Maria Domenica Di Benedetto and Alberto L. Sangiovanni{-}Vincentelli},
  title = {Minimum-Cost Reachability for Priced Timed Automata},
  booktitle = {Hybrid Systems: Computation and Control, 4th International Workshop, {HSCC} 2001, Rome, Italy, March 28-30, 2001, Proceedings},
  year = {2001},
  publisher = {Springer},
  address = {Berlin, Heidelberg},
  pages = {147--161},
  abstract = {This paper introduces the model of linearly priced timed automata as an extension of timed automata, with prices on both transitions and locations. For this model we consider the minimum-cost reachability problem: i.e. given a linearly priced timed automaton and a target state, determine the minimum cost of executions from the initial state to the target state. This problem generalizes the minimum-time reachability problem for ordinary timed automata. We prove decidability of this problem by offering an algorithmic solution, which is based on a combination of branch-and-bound techniques and a new notion of priced regions. The latter allows symbolic representation and manipulation of reachable states together with the cost of reaching them.},
  isbn = {978-3-540-45351-2},
  timestamp = {Mon, 03 Jan 2022 00:00:00 +0100},
  biburl = {https://dblp.org/rec/conf/hybrid/BehrmannFHLPRV01.bib},
  bibsource = {dblp computer science bibliography, https://dblp.org},
  doi = {10.1007/3-540-45351-2_15},
  volume = {2034},
  url = {https://doi.org/10.1007/3-540-45351-2\_15},
  series = {Lecture Notes in Computer Science},
  _bib2doi_selected = {dblp:/rec/conf/hybrid/BehrmannFHLPRV01.bib},
  _bib2doi_confirmed = {true},
  _bib2doi_finished = {true},
}

@article{BGHM16,
  title = {{Pseudopolynomial iterative algorithm to solve total-payoff games and min-cost reachability games}},
  author = {Brihaye, Thomas and Geeraerts, Gilles and Haddad, Axel and Monmege, Benjamin},
  url = {https://hal.science/hal-01414114},
  journal = {{Acta Informatica}},
  publisher = {{Springer Verlag}},
  series = {Special Issue: Selected papers from the 26th International Conference on Concurrency Theory (CONCUR 2015)},
  volume = {54},
  number = {1},
  pages = {85--125},
  year = {2017},
  month = {feb},
  doi = {10.1007/s00236-016-0276-z},
  pdf = {https://hal.science/hal-01414114v1/file/priced-games.pdf},
  hal_id = {hal-01414114},
  hal_version = {v1},
  timestamp = {Sun, 21 Jun 2020 01:00:00 +0200},
  biburl = {https://dblp.org/rec/journals/acta/BrihayeGHM17.bib},
  bibsource = {dblp computer science bibliography, https://dblp.org},
  _bib2doi_selected = {dblp:/rec/journals/acta/BrihayeGHM17.bib},
  _bib2doi_confirmed = {true},
}

@inproceedings{guilmant2024,
  author = {Guilmant, Quentin and Ouaknine, Jo\"{e}l},
  title = {{Inaproximability in Weighted Timed Games}},
  booktitle = {35th International Conference on Concurrency Theory (CONCUR 2024)},
  pages = {27:1--27:15},
  series = {Leibniz International Proceedings in Informatics (LIPIcs)},
  isbn = {978-3-95977-339-3},
  issn = {1868-8969},
  year = {2024},
  volume = {311},
  editor = {Majumdar, Rupak and Silva, Alexandra},
  publisher = {Schloss Dagstuhl -- Leibniz-Zentrum f{\"u}r Informatik},
  address = {Dagstuhl, Germany},
  url = {https://drops.dagstuhl.de/entities/document/10.4230/LIPIcs.CONCUR.2024.27},
  urn = {urn:nbn:de:0030-drops-207998},
  doi = {10.4230/LIPIcs.CONCUR.2024.27},
  annote = {Keywords: Weighted timed games, approximation, undecidability},
  timestamp = {Fri, 30 Aug 2024 01:00:00 +0200},
  biburl = {https://dblp.org/rec/conf/concur/GuilmantO24.bib},
  bibsource = {dblp computer science bibliography, https://dblp.org},
  _bib2doi_selected = {dblp:/rec/conf/concur/GuilmantO24.bib},
  _bib2doi_confirmed = {true},
}

@inproceedings{bouyer2006,
  author = {Patricia Bouyer and Kim Guldstrand Larsen and Nicolas Markey and Jacob Illum Rasmussen},
  editor = {S. Arun{-}Kumar and Naveen Garg},
  title = {Almost Optimal Strategies in One Clock Priced Timed Games},
  booktitle = {{FSTTCS} 2006: Foundations of Software Technology and Theoretical Computer Science, 26th International Conference, Kolkata, India, December 13-15, 2006, Proceedings},
  year = {2006},
  publisher = {Springer},
  address = {Berlin, Heidelberg},
  pages = {345--356},
  abstract = {We consider timed games extended with cost information, and prove computability of the optimal cost and of $\epsilon$-optimal memoryless strategies in timed games with one clock. In contrast, this problem has recently been proved undecidable for timed games with three clocks.},
  isbn = {978-3-540-49995-4},
  timestamp = {Wed, 24 May 2017 01:00:00 +0200},
  biburl = {https://dblp.org/rec/conf/fsttcs/BouyerLMR06.bib},
  bibsource = {dblp computer science bibliography, https://dblp.org},
  doi = {10.1007/11944836_32},
  volume = {4337},
  url = {https://doi.org/10.1007/11944836\_32},
  series = {Lecture Notes in Computer Science},
  _bib2doi_selected = {dblp:/rec/conf/fsttcs/BouyerLMR06.bib},
  _bib2doi_confirmed = {true},
  _bib2doi_finished = {true},
}

@inproceedings{Bouyer2015,
  author = {Bouyer, Patricia and Jaziri, Samy and Markey, Nicolas},
  title = {{On the Value Problem in Weighted Timed Games}},
  booktitle = {26th International Conference on Concurrency Theory (CONCUR 2015)},
  pages = {311--324},
  series = {Leibniz International Proceedings in Informatics (LIPIcs)},
  isbn = {978-3-939897-91-0},
  issn = {1868-8969},
  year = {2015},
  volume = {42},
  editor = {Aceto, Luca and de Frutos Escrig, David},
  publisher = {Schloss Dagstuhl -- Leibniz-Zentrum f{\"u}r Informatik},
  address = {Dagstuhl, Germany},
  url = {https://drops.dagstuhl.de/entities/document/10.4230/LIPIcs.CONCUR.2015.311},
  urn = {urn:nbn:de:0030-drops-53863},
  doi = {10.4230/LIPIcs.CONCUR.2015.311},
  annote = {Value problem undecidable for WTG. Approximation for almost-divergent WTG on N},
  timestamp = {Thu, 23 Aug 2018 01:00:00 +0200},
  biburl = {https://dblp.org/rec/conf/concur/BouyerJM15.bib},
  bibsource = {dblp computer science bibliography, https://dblp.org},
  _bib2doi_selected = {dblp:/rec/conf/concur/BouyerJM15.bib},
  _bib2doi_confirmed = {true},
}

@article{BBL2008,
  author = {Bouyer, Patricia and Brinksma, Ed and Larsen, Kim G.},
  title = {Optimal infinite scheduling for multi-priced timed automata},
  year = {2008},
  booktitle = {Formal methods in system design},
  volume = {32},
  pages = {3-23},
  timestamp = {Tue, 07 May 2024 01:00:00 +0200},
  biburl = {https://dblp.org/rec/journals/fmsd/BouyerBL08.bib},
  bibsource = {dblp computer science bibliography, https://dblp.org},
  doi = {10.1007/s10703-007-0043-4},
  _bib2doi_selected = {dblp:/rec/journals/fmsd/BouyerBL08.bib},
  _bib2doi_confirmed = {true},
}

@inproceedings{Busatto2017,
  author = {Busatto-Gaston, Damien and Monmege, Benjamin and Reynier, Pierre-Alain},
  editor = {Esparza, Javier and Murawski, Andrzej S.},
  title = {Optimal Reachability in Divergent Weighted Timed Games},
  booktitle = {Foundations of Software Science and Computation Structures},
  year = {2017},
  publisher = {Springer Berlin Heidelberg},
  address = {Berlin, Heidelberg},
  pages = {162--178},
  abstract = {Weighted timed games are played by two players on a timed automaton equipped with weights: one player wants to minimise the accumulated weight while reaching a target, while the other has an opposite objective. Used in a reactive synthesis perspective, this quantitative extension of timed games allows one to measure the quality of controllers. Weighted timed games are notoriously difficult and quickly undecidable, even when restricted to non-negative weights. Decidability results exist for subclasses of one-clock games, and for a subclass with non-negative weights defined by a semantical restriction on the weights of cycles. In this work, we introduce the class of divergent weighted timed games as a generalisation of this semantical restriction to arbitrary weights. We show how to compute their optimal value, yielding the first decidable class of weighted timed games with negative weights and an arbitrary number of clocks. In addition, we prove that divergence can be decided in polynomial space. Last, we prove that for untimed games, this restriction yields a class of games for which the value can be computed in polynomial time.},
  isbn = {978-3-662-54458-7},
  timestamp = {Tue, 07 May 2024 01:00:00 +0200},
  biburl = {https://dblp.org/rec/conf/fossacs/Busatto-GastonM17.bib},
  bibsource = {dblp computer science bibliography, https://dblp.org},
  doi = {10.1007/978-3-662-54458-7_10},
  _bib2doi_selected = {dblp:/rec/conf/fossacs/Busatto-GastonM17.bib},
  _bib2doi_confirmed = {true},
}

@inproceedings{Busatto2018,
  author = {Busatto-Gaston, Damien and Monmege, Benjamin and Reynier, Pierre-Alain},
  title = {{Symbolic Approximation of Weighted Timed Games}},
  booktitle = {38th IARCS Annual Conference on Foundations of Software Technology and Theoretical Computer Science (FSTTCS 2018)},
  pages = {28:1--28:16},
  series = {Leibniz International Proceedings in Informatics (LIPIcs)},
  isbn = {978-3-95977-093-4},
  issn = {1868-8969},
  year = {2018},
  volume = {122},
  editor = {Ganguly, Sumit and Pandya, Paritosh},
  publisher = {Schloss Dagstuhl -- Leibniz-Zentrum f{\"u}r Informatik},
  address = {Dagstuhl, Germany},
  url = {https://drops.dagstuhl.de/entities/document/10.4230/LIPIcs.FSTTCS.2018.28},
  urn = {urn:nbn:de:0030-drops-99277},
  doi = {10.4230/LIPIcs.FSTTCS.2018.28},
  annote = {Keywords: Weighted timed games, Real-time systems, Game theory, Approximation},
  timestamp = {Tue, 07 May 2024 01:00:00 +0200},
  biburl = {https://dblp.org/rec/conf/fsttcs/Busatto-GastonM18.bib},
  bibsource = {dblp computer science bibliography, https://dblp.org},
  _bib2doi_selected = {dblp:/rec/conf/fsttcs/Busatto-GastonM18.bib},
  _bib2doi_confirmed = {true},
}

@article{Alur94,
  title = {A Theory of Timed Automata},
  journal = {Theor. Comput. Sci.},
  volume = {126},
  number = {2},
  pages = {183--235},
  year = {1994},
  issn = {0304-3975},
  doi = {10.1016/0304-3975(94)90010-8},
  url = {https://doi.org/10.1016/0304-3975(94)90010-8},
  author = {Rajeev Alur and David L. Dill},
  abstract = {We propose timed (finite) automata to model the behavior of real-time systems over time. Our definition provides a simple, and yet powerful, way to annotate state-transition graphs with timing constraints using finitely many real-valued clocks. A timed automaton accepts timed words–infinite sequences in which a real-valued time of occurrence is associated with each symbol. We study timed automata from the perspective of formal language theory: we consider closure properties, decision problems, and subclasses. We consider both nondeterministic and deterministic transition structures, and both Büchi and Muller acceptance conditions. We show that nondeterministic timed automata are closed under union and intersection, but not under complementation, whereas deterministic timed Muller automata are closed under all Boolean operations. The main construction of the paper is an (PSPACE) algorithm for checking the emptiness of the language of a (nondeterministic) timed automaton. We also prove that the universality problem and the language inclusion problem are solvable only for the deterministic automata: both problems are undecidable (Π11-hard) in the nondeterministic case and PSPACE-complete in the deterministic case. Finally, we discuss the application of this theory to automatic verification of real-time requirements of finite-state systems.},
  annote = {Sec 4: Region construction for timed automata},
  timestamp = {Wed, 17 Feb 2021 00:00:00 +0100},
  biburl = {https://dblp.org/rec/journals/tcs/AlurD94.bib},
  bibsource = {dblp computer science bibliography, https://dblp.org},
  _bib2doi_old_doi = {https://doi.org/10.1016/0304-3975(94)90010-8},
  _bib2doi_selected = {dblp:/rec/journals/tcs/AlurD94.bib},
  _bib2doi_confirmed = {true},
  _bib2doi_finished = {true},
}

@article{monmege2024journal,
    title      = {Decidability of One-Clock Weighted Timed Games with Arbitrary Weights},
    author     = {Benjamin Monmege and Julie Parreaux and Pierre-Alain Reynier},
    url        = {https://lmcs.episciences.org/11129},
    doi        = {10.46298/lmcs-21(1:8)2025},
    journal    = {Logical Methods in Computer Science},
    issn       = {1860-5974},
    volume     = {Volume 21, Issue 1},
    eid        = 8,
    year       = {2025},
    month      = {Jan},
}

@misc{two-clock-undec,
  title = {The Value Problem for Weighted Timed Games with Two Clocks is Undecidable},
  author = {Quentin Guilmant and Joël Ouaknine and Isa Vialard},
  year = {2025},
  note = {accepted at FOSSACS 2026},
  eprint = {2507.10550},
  archiveprefix = {arXiv},
  primaryclass = {cs.GT},
  url = {https://arxiv.org/abs/2507.10550},
  doi = {10.48550/arXiv.2507.10550},
  timestamp = {Sat, 06 Sep 2025 01:00:00 +0200},
  biburl = {https://dblp.org/rec/journals/corr/abs-2507-10550.bib},
  bibsource = {dblp computer science bibliography, https://dblp.org},
  _bib2doi_selected = {dblp:/rec/journals/corr/abs-2507-10550.bib},
  _bib2doi_confirmed = {true},
}

@misc{vialard2025decidingvaluetwoclocknonzeno,
  title = {Deciding the Value of Two-Clock Almost Non-Zeno Weighted Timed Games},
  author = {Isa Vialard},
  year = {2025},
  eprint = {2508.00014},
  archiveprefix = {arXiv},
  primaryclass = {cs.LO},
  url = {https://arxiv.org/abs/2508.00014},
  note = {Accepted at CSL 2026},
  doi = {10.48550/arXiv.2508.00014},
  timestamp = {Wed, 10 Sep 2025 01:00:00 +0200},
  biburl = {https://dblp.org/rec/journals/corr/abs-2508-00014.bib},
  bibsource = {dblp computer science bibliography, https://dblp.org},
  _bib2doi_selected = {dblp:/rec/journals/corr/abs-2508-00014.bib},
  _bib2doi_confirmed = {true},
}

\clearpage



\end{document}